\documentclass[11pt]{amsart}
\usepackage[latin1]{inputenc}
\usepackage[all]{xy}
\usepackage{amssymb, amsmath, amscd, geometry}
\usepackage{graphicx}
\usepackage{setspace}
\numberwithin{equation}{section}
\input xy
\xyoption{all}
\usepackage{latexsym}
\usepackage[usenames]{color}
\usepackage{epic} 
\usepackage{mathrsfs}
\usepackage{euscript}
\usepackage{rotating}
\usepackage{bbm}

\usepackage{verbatim}

\numberwithin{equation}{section}
\input xy
\xyoption{all}
\usepackage{latexsym}
\usepackage[usenames]{color}
\usepackage{epic} 
\usepackage{mathrsfs}
\usepackage{euscript}
\usepackage{rotating}
\usepackage{bbm}

\usepackage{verbatim}
\usepackage[usenames]{color}

\theoremstyle{plain}
\newtheorem{lemma}{Lemma}[section]
\newtheorem{theorem}[lemma]{Theorem}
\newtheorem{corollary}[lemma]{Corollary}

\newtheorem{prop}[lemma]{Proposition}

\theoremstyle{definition}

\newtheorem{remark}{Remark}[section]

\newtheorem*{theorem-non}{Theorem}

\newcommand{\C}{\ensuremath{\mathbb{C}}}

\newcommand{\R}{\ensuremath{\mathbb{R}}}

\newcommand{\pe}{\otimes_\epsilon}
\newcommand{\setx}{\textbf{X}}
\newcommand{\cardx}{\textsc{x}}
\newcommand{\sety}{\textbf{Y}}
\newcommand{\cardy}{\textsc{y}}
\newcommand{\setz}{\textbf{Z}}

\newcommand{\seta}{\textbf{A}}
\newcommand{\carda}{\textsc{a}}
\newcommand{\setb}{\textbf{B}}
\newcommand{\cardb}{\textsc{b}}

\newcommand{\lra}{\longrightarrow}

\newcommand{\vr}{\varepsilon}

\newcommand{\sign}{\operatorname{sign}}

              \newcommand{\Id}{\operatorname{Id}}

\newcommand{\uno}{\mathbbm{1}}

\allowdisplaybreaks

\begin{document}


\title{Classical vs. quantum communication in XOR games}

\author{Marius Junge}
\address{Department of Mathematics\\University of Illinois at Urbana-Champaign\\1409 W. Green St. Urbana, IL 61891. USA}
\email{junge@math.uiuc.edu}

\author{Carlos Palazuelos}
\address{Instituto de Ciencias Matem�ticas (ICMAT)\\Departamento de An\'alisis Matem\'atico \\
Facultad de Ciencias Matem\'aticas \\ Universidad Complutense de Madrid \\
Madrid 28040. Spain}
\email{cpalazue@mat.ucm.es}

\author{Ignacio Villanueva} 
\address{Departamento de An\'alisis Matem\'atico \\
Facultad de Ciencias Matem\'aticas \\ Universidad Complutense de Madrid \\
Madrid 28040. Spain}
\email{ignaciov@mat.ucm.es}

\thanks{The first author is partially supported by NSF DMS-1201886. The second author is partially supported by the Spanish ``Ram\'on y Cajal program'' (RYC-2012-10449). The first and second author are partially supported by the Spanish ``Severo Ochoa Programe'' for Centres of Excellence (SEV-2015-0554). The second and third authors are partially supported by the grants MTM2014-54240-P, funded by Spanish MINECO and QUITEMAD+-CM, S2013/ICE-2801, funded by Comunidad de Madrid.}

\begin{abstract}
In this work we introduce an intermediate setting between quantum nonlocality and communication complexity problems. More precisely, we study the value of XOR games $G$ when Alice and Bob are allowed to use a limited amount of one-way classical communication  $\omega_{o.w.-c}(G)$ (resp. one-way quantum communication $\omega_{o.w.-c}^*(G)$), where $c$ denotes the number of bits (resp. qubits). The key quantity here is the quotient $\omega_{o.w.-c}^*(G)/\omega_{o.w.-c}(G)$.

We provide a universal way to obtain Bell inequality violations of general Bell functionals from XOR games for which the quotient $\omega_{o.w.-c}^*(G)/\omega_{o.w.-2c}(G)$ is larger than 1. This allows, in particular, to find (unbounded) Bell inequality violations from communication complexity problems in the same spirit as the recent work by Buhrman et al. (2016).

We also provide an example of a XOR game for which the previous quotient is optimal (up to a logarithmic factor) in terms of the amount of information $c$. Interestingly, this game has only polynomially many inputs per player. For the related problem of separating the classical vs quantum communication complexity of a function, the known examples attaining exponential separation require exponentially many inputs per party.
\end{abstract}

\maketitle

\section{Introduction and main results}

One of the main themes of research in Quantum Information is the quantification of the advantages provided from the use of quantum resources versus the use of classical resources. This quantification has been studied in many different contexts, the first one historically being Bell inequalities.  In a Bell experiment \cite{Bell}, Alice and Bob perform some measurements indexed by $x\in \setx$, $y\in\sety$ on a bipartite system and obtain some outputs $a\in \seta$, $b\in \setb$ respectively. The repetition of the experiment a large number of times leads to a bipartite probability distribution $P=(P(a,b|x,y))_{a\in \seta, b\in \setb, x\in \setx, y\in \sety}$. A standard way to quantify the phenomenon of quantum nonlocality is to consider linear functionals acting on the distribution $P$ considered as an element of $\mathbb R^{\cardx\cardy\carda\cardb}$, where $\cardx$, $\cardy$, $\carda$, $\cardb$ denote the cardinal of $\setx$, $\sety$, $\seta$ and $\setb$ respectively. A general linear functional $B$ on $\mathbb R^{\cardx\cardy\carda\cardb}$ is given by a collection of numbers $(B_{a,b,x,y})_{a\in \seta, b\in \setb, x\in \setx, y\in \sety}$ and its action on a given probability distribution $P$ is defined as $$\langle B, P\rangle=\sum_{a,b,x,y} B_{a,b,x,y} P(a,b|x,y).$$ We will refer to one such $B$ as a  {\em Bell functional}. Then, one can define $$\omega(B)=\sup_{P\in \mathcal L}\left|\langle B, P\rangle\right|\text{      }\text{      }\text{  and    }\text{      }\text{      }\omega^*(B)=\sup_{P\in \mathcal Q}\left|\langle B, P\rangle\right|,$$where  $\mathcal L$ is the set of classical bipartite probability distributions and $\mathcal Q$ is the set of quantum  bipartite probability distributions; that is, those probability distributions that Alice and Bob can generate when they share an unlimited amount of entanglement. The key ratio to quantify quantum nonlocality is $\omega^*(B)/\omega(B)$ and we say that there exists a Bell inequality violation if this quotient is strictly larger than 1.

Another relevant context where quantum resources perform better than classical resources is communication complexity.  In the usual task  \cite{NiKu} two separate parties, Alice and Bob, have to compute a binary function $f(x,y)$ of two predicates $x\in \setx$, $y\in \sety$. Alice only has access to $x$, whereas Bob only has access to $y$. They are assumed to have unlimited computational resources, and they can interchange messages until they are able to compute the function. The randomized communication complexity of $f$ is defined to be the minimum number of bits (or qubits in the quantum case) interchanged between Alice and Bob required for a randomized algorithm  in order to compute correctly $f(x,y)$ with probability larger than $\epsilon$ for every possible input $(x,y)$. We will call these numbers $CC(f,\epsilon)$ and $QC(f,\epsilon)$ respectively.  

In this paper we study the relation between Bell inequality violations and communication complexity problems \cite{BCMW10}, continuing the spirit of the recent paper \cite{Buhrman16}, where some new implications between both contexts where uncovered. In this line, certain specific Bell inequality violations are known to lead to separation in communication complexity for certain functions \cite{BZPZ}, although we do not know of any general implication in this direction. For the other direction, a recent result  \cite{Buhrman16}  shows that it is possible in general to obtain Bell inequality violations starting from large enough separations in communication complexity. Some other interesting results relating communication complexity and Bell inequality violations have been recently obtained in \cite{LLNRS}. However, in this last work the authors study \emph{inefficiency-resistant} Bell inequalities, which is a different notion from the one studied in our paper.

The novelty of our approach is to introduce an intermediate setting between quantum nonlocality and communication complexity problems. More precisely, we study the value of two-prover one-round games when Alice and Bob are allowed to communicate a limited amount of information.  This is related to the {\em distributional complexity} studied in communication complexity, with the difference that in the communication complexity context typically one fixes the desired probability of winning and calculates the amount of communication needed, whereas we will fix the amount of communication and calculate the probability of winning for that amount of communication.  In our study, we will restrict ourselves all the time to  the case of {\em one-way communication}, that is, the communication can only be sent from Alice to Bob and not the other way around. 

In this work, we study this situation for XOR games. These games can be understood as particular Bell functional where the set of outputs are $\seta=\setb=\{-1,1\}$ and the coefficients of the functional have the form
\begin{align}\label{coefficients XOR games}
G_{a,b,x,y}=\pi(x,y)ab f(x,y).
\end{align}Here $(\pi(x,y))_{x,y}$ is a probability distribution and $f:\setx\times \sety\rightarrow \{-1,1\}$ is any function. Any XOR game\footnote{XOR games are often introduced replacing $\{-1,1\}$ by $\{0,1\}$ and replacing the product by the XOR of the variables. In that case, our {\em value} of the game translates into the {\em bias}, the additional probability over $\frac{1}{2}$ of winning the game.} is then uniquely determined by $\pi$ and $f$, so we will denote it by  $G=(f,\pi)$. We usually use notation $G$ for XOR games, compared to general Bell functionals which we will denote by $B$. In addition to the quantities defined above, we will also consider the quantities  
$$\omega_{o.w.-c}(B)=\sup_{P\in \mathcal OW_c}\left|\langle B, P\rangle\right|\text{      }\text{      }\text{  and    }\text{      }\text{      }\omega^*_{o.w.-c}(B)=\sup_{P\in \mathcal OW^*_c}\left|\langle B, P\rangle\right|,$$where here $\mathcal OW_c$ is the set of bipartite probability distributions that Alice and Bob can generate using classical resources and $c$ bits of classical communication sent from Alice to Bob and 
$\mathcal OW^*_c$ is the set of bipartite probability distributions that Alice and Bob can generate using classical resources and $c$ qubits of  communication sent from Alice to Bob (they are not allowed to share additional entanglement).

A clever application of Grothendieck's inequality allowed Tsirelson to show that $\omega^*(G)/\omega(G)\leq K_G^{\R}\leq  1.7822\cdots $, where $K_G^{\R}$ is the so called \emph{real Grothendieck's constant}, and this inequality holds for every XOR game $G$ (independently of the number of inputs) \cite{Tsirelson}. In contrast to this,  the quantity
\begin{align}\label{quotient w*-w}
\frac{\omega^*_{o.w.-c}(G)}{\omega_{o.w.-c}(G)}
\end{align} for XOR games can lead to unbounded quotients and is, in general,  highly non trivial to estimate.  Our main results involve  the quantity (\ref{quotient w*-w}) and find applications in other settings. 

\subsection*{Bell inequality violations arising from advantages in the quantum vs classical communication value of XOR games}

The first contribution of this work is to show that there exists a universal way to obtain Bell inequality violations  from XOR games for which the quantity (\ref{quotient w*-w}) is strictly larger than 1. As we explain below this provides a universal way to obtain Bell inequality violations from quantum vs. classical advantages in communication complexity.

More precisely, given a XOR game $G=(\pi,f)$ with coefficients $T_{x,y}=\pi(x,y)f(x,y)$ for every $x\in \setx$, $y\in\sety$, and given a natural number $d$, we will consider a Bell functional $B^G_d$ defined as follows:

- Set of inputs for Alice and Bob: $\tilde{\setx}=\setx$ and $\tilde{\sety}=\sety\times \{1,\cdots, d\}$, respectively.

- Set of outputs for Alice and Bob: $\tilde{\seta}=\{1,\cdots, d\}\times \{-1,1\}$ and $\tilde{\setb}=\{-1,1\}$, respectively.

The coefficients of the Bell functional are defined, for every $x\in \setx$, $(y,k)\in \sety\times \{1,\cdots, d\}$, $(a,\tilde{a})\in \{1,\cdots, d\}\times \{-1,1\}$, $b\in \{-1,1\}$, as:
\begin{align}\label{bell functional definition}
B^G_d(a,\tilde{a}, b,x,y,k)=T_{x,y}\cdot \delta_{a,k} \cdot \tilde{a}\cdot b,
\end{align}where $\delta_{a,k}$ equals one if $a=k$ and equals zero otherwise.
\begin{theorem}\label{Thm:MainI}
Let $G$ be a XOR game. With the previous notation, we have   $$\frac{\omega^*(B^G_{d^2})}{\omega(B^G_{d^2})}\geq \frac{\omega^*_{o.w-\log d}(G)}{\omega_{o.w-2\log d}(G)}.$$ 
\end{theorem}

It is known that applying the min-max theorem of zero sum games and Chernoff bound one can prove  statements of the following type: If there is a function $f$ for which $CC(f,\epsilon)$ is sufficiently larger than $QC(f,\epsilon)=c$, then there is a probability distribution $\pi$ such that the associated XOR game $G=(f,\pi)$ verifies
that $\omega^*_{o.w.-c}$ is larger than  $\omega_{o.w.-c}$. The precise statement we will need, and its proof, will be given in Section \ref{S:BV}. That is, starting from a separation in the Classical vs. Quantum communication complexity of a function $f$, we can generate a XOR game $G=(f, \pi)$ with separation in the value with $c$ qubits of communication vs the value with $c$ bits of classical communication. 
Making use of this connection and our Theorem \ref{Thm:MainI}, we will show the following result.
\begin{theorem}\label{Thm:MainII}
Let $\alpha>216$ and let $f:\setx\times \sety \longrightarrow \{-1, 1\}$ be a function for which $$\frac{CC(f,\frac{2}{3})}{QC(f,\frac{2}{3})}> \alpha.$$

Then, there exist a probability distribution $\pi:\setx\times \sety\longrightarrow [0,1]$ and an associated XOR game $G=(f,\pi)$ which verify the following: Calling $c=QC(f,\frac{2}{3})$ and $d=2^c$, the Bell functional $B_{d^2}^G$ defined just before Theorem \ref{Thm:MainI} verifies that
\begin{align*}
\frac{\omega^*(B_{d^2}^G)}{\omega(B_{d^2}^G)}\geq
\frac{\sqrt{\alpha}}{6\sqrt{6}}.
\end{align*}
\end{theorem}

As stated before, both  Theorem \ref{Thm:MainI} and Theorem \ref{Thm:MainII} follow the spirit of \cite{Buhrman16}. We comment on  how our results compares with the main result in \cite{Buhrman16}. Our results need less advantage in the communication complexity and attain a bigger violation with an arguably simpler Bell functional. In particular, the main result in \cite{Buhrman16} is the derivation of Bell inequality violations from advantage in communication complexity, but in that paper it is required that $CC(f,\frac{2}{3})\geq QC(f,\frac{2}{3})^4$, whereas we only require $CC(f,\frac{2}{3})\geq \alpha QC(f,\frac{2}{3})$ for constant $\alpha$. On the down side, our results apply only to one-way communication complexity, whereas the result in \cite{Buhrman16} applies to general (two way) communication complexity (see Sections \ref{S:Bell violations} and  \ref{S:BV} for details). Whether our techniques can be extended to the case of two-way communication between the players remains as an open problem. 
\subsection*{An example of a XOR game $G$ essentially achieving the maximal possible separation between $\omega^*_{o.w.-c}(G)$ and  $\omega_{o.w.-c}(G)$.}

Our second contribution is the study of how much advantage can be obtained using quantum communication versus classical communication in the maximum value of a XOR game with restricted amount of one-way communication.  
\begin{theorem}\label{Thm:MainIII}
There exists a XOR game $G$ with $2^{2n}$ inputs for Alice and $2^{n^2}$ inputs for Bob such that for every $k\geq e^2$ we have the following estimate: $$\frac{\omega^*_{o.w-\log n}(G)}{\omega_{o.w-\log k}(G)}\geq C\frac{\sqrt{n}}{\log k},$$where $C$ is a universal constant.

Moreover, this statement is essentially optimal in the sense that for every XOR game $G$ and for every $k\geq 1$ we have
$$\frac{\omega^*_{o.w-\log n}(G)}{ \omega_{o.w-\log k}(G)}\leq \frac{\omega^*_{o.w-\log n}(G)}{ \omega(G)}\leq K_G^{\mathbb R}\sqrt{n},$$ where here $K_G^{\mathbb R}$ denotes the real Grothendieck's constant and $\omega(G)$ is the classical value of $G$ with no communication ($k=1$).
\end{theorem}

An immediate consequence of Theorem \ref{Thm:MainI} and Theorem \ref{Thm:MainIII} is that we can obtain unbounded Bell inequality violations with a very particular kind of Bell functionals $B$. Indeed, the previous results applied to $k=n^2$ imply that there exists a Bell functional $B$ with $2^{2n}$ inputs for Alice,  $n^22^{n^2} $ inputs for Bob,  $2n^2$ outputs for Alice and $2$ outputs for Bob such that the quotient of $\omega^*(B)/\omega(B)$ is $\Omega\big(\sqrt{n}/\log n\big)$.

It is well known that Bell functionals with only two outputs per player cannot give large violations. Indeed, one can show (\cite{JP, PV-Survey}) that $\omega^*(B)/\omega(B)=O(\min\{N,K\})$ for any Bell functional $B$, where $N=\min\{\cardx, \cardy\}$ and $K=\sqrt{\carda\cardb}$. Hence, Bell functionals with dichotomic outputs for one player can be seen as the simplest possible example leading to large violations. Previous examples of Bell functionals showing these kinds of behaviors where shown in \cite{PaYi}, but the functional $B$ defined in (\ref{bell functional definition}) has an even simpler form than the one considered in \cite{PaYi}. 
\subsection*{A scale improvement over previously known examples}

The very particular structure of our above mentioned game $G$ will allow us to apply an  atom reduction method developed in \cite{JOP16} to obtain another XOR game with polynomially many inputs and attaining the same Bell inequality violation. Up to our knowledge,  all previous examples of this flavor used exponentially many inputs per party.
\begin{theorem}\label{Thm:MainIV}
There exists a XOR game $G$ with $cn^8$ inputs for Alice and Bob such that   $$\frac{\omega^*_{o.w-\log n}(G)}{\omega_{o.w-\log n}(G)}\geq C\frac{\sqrt{n}}{\log n}.$$Here, $C$ and $c$ are universal constants.
\end{theorem}

In fact, a slight modification of Theorem \ref{Thm:MainIV} together with Theorem \ref{Thm:MainI}  allows us to obtain:
\begin{corollary}\label{Corollary last}
There exists a Bell functional $B$ with $cn^{12}$ and $cn^{14}$ inputs and $2n^2$ and $2$ outputs for Alice and Bob respectively such that $$\frac{\omega^*(B)}{\omega(B)}\geq D'\frac{\sqrt{n}}{\log n},$$where $D'$ and $c$ are universal constants.
\end{corollary}

 The paper is organized as follows. In Section \ref{S: communication} we introduce the basic definitions and we explain in detail the communication complexity models that we will use in the rest of the paper. In particular, we define properly the values $\omega_{o.w-\log n}(G)$ and $\omega^*_{o.w-\log n}(G)$ for XOR games. At the end of the section we provide an upper bound for the quotient $\omega^*_{o.w-\log n}(G)/\omega_{o.w-\log n}(G)$ (Proposition  \ref{prop optimality}) which proves the optimality part of Theorem \ref{Thm:MainIII}. In Section \ref{S:Bell violations} we prove Theorem \ref{Thm:MainI}. In Section \ref{S:BV} we explain how to obtain (large) Bell inequality violations from a separation in the classical vs. quantum communication complexity of a boolean function $f$ by proving Theorem \ref{Thm:MainII}. In Section \ref{S:values} we introduce a new XOR game $G$ and we prove the lower bound in Theorem \ref{Thm:MainIII}. In Section \ref{S:tensor norms} we explain that the study performed in this work can be understood in terms of tensor norms and, in particular, we express all the quantities introduced in the previous sections in terms of certain (well known) tensor norms. This point of view is crucial in Section \ref{S:atom reduction}, where we use the new language to reduce the number of inputs of the game defined in Section \ref{S:values} while preserving the quotient $\omega^*_{o.w-\log n}(G)/\omega_{o.w-\log n}(G)$. That is, we prove Theorem \ref{Thm:MainIV}.
\section{Communication complexity models}\label{S: communication}

In this section, we describe the mathematical models associated to the value of a XOR game when the players are assisted with $c$ classical or quantum bits. At the end of the section, we bound the maximal possible difference between both cases.  

Before we start describing these models, we will state the precise definition of the classical and the quantum value of a Bell functional. Given a Bell functional $B=(B_{a,b,x,y})_{a\in \seta, b\in \setb, x\in \setx, y\in \sety}$ we define:

- \noindent the \emph{classical value} of $B$ as
\begin{align}\label{classical value}
\omega(B)=\sup_{P\in \mathcal L} |\langle B, P\rangle|,
\end{align}where $\mathcal L$ is the set of probability distributions of the form
\begin{equation*}\label{classical}
P(a,b|x,y)=\int_\Omega P_\omega(a|x)Q_\omega(b|y)d\mathbb{P}(\omega)
\end{equation*}
for every $x,y,a,b$. Here, $(\Omega,\Sigma,\mathbb{P})$ is a probability space, $P_\omega(a|x)\ge 0$ for all $a,x,\omega$, $\sum_a
P_\omega(a|x)=1$ for all $x,\omega$, and the analogous conditions hold for $Q_\omega(b|y)$. 

-\noindent the \emph{quantum value} of $B$ as
\begin{align}\label{quantum value}
\omega^*(B)=\sup_{P\in \mathcal Q}|\langle B, P\rangle|,
\end{align}where $\mathcal Q$ is the set of probability distribution of the form
\begin{equation*}
P(a,b|x,y)=tr(E_x^a\otimes F_y^b \rho)
\end{equation*}
for every $x,y,a,b$. Here $\rho$ is a density operator acting on the tensor product of two Hilbert spaces $H_1\otimes H_2$ and $(E_x^a)_{x,a}$ and $(F_y^b)_{y,b}$ are two sets
of operators representing POVM
measurements acting on $H_1$ and $H_2$ respectively. That is, $E_x^a\geq 0$ for every  $x,a$, $\sum
_{a}E_x^a=\Id$ for every $x$, and the analogous conditions hold for $(F_y^b)_{y,b}$. 

An interesting measure to quantify nonlocality is then \emph{the  Bell violation of the functional $B$}, $\omega^*(B)/\omega(B)$. This magnitude has been deeply studied in the last years from the point of view of physics and computer sciences (since Bell functionals can be associated to two-prover one-round games) \cite{JP, PV-Survey}.

Note that for the particular case of a XOR game $G=(\pi, f)$ as in (\ref{coefficients XOR games}), for any bipartite probability distribution $P$ we have 
\begin{align}\label{bias}
\langle G, P\rangle&=\sum_{x\in \setx,y\in \sety} \pi(x,y)\big[ P(ab=f(x,y)|x,y)- P(ab=-f(x,y)|x,y)\big]\\&\nonumber=
\sum_{x\in \setx,y\in \sety} \pi(x,y)f(x,y) \mathbb E(ab|x,y).
\end{align}

This motivates us to consider the correlation matrix $(\gamma_{x,y})_{x,y}=\mathbb E(ab|x,y)$ associated to the strategy $P$ and to biunivocally characterize the XOR game $G$ by the coefficients  $T_{x,y}=\pi(x,y) f(x,y)$, so that 
\begin{align*}
\langle G, P\rangle =\sum_{x\in \setx, y\in\sety}T_{x,y}\gamma_{x,y}.
\end{align*}

It is very easy to see from the definitions above that for a given XOR game $G$ with coefficients $T_{x,y}=\pi(x,y) f(x,y)$ one has 
\begin{align}\label{classical bias XOR}
\omega(G)=\sup\Big|\sum_{x\in \setx, y\in\sety}T_{x,y}\int_\Omega A_\omega(x)B_\omega(y)d\mathbb{P}(\omega)\Big|=\sup\Big|\sum_{x\in \setx, y\in\sety}T_{x,y}t_xs_y\Big|,
\end{align}where the first supremum runs over all probability spaces $(\Omega,\Sigma,\mathbb{P})$ and all families of real numbers $(A_\omega(x))_x$, $(B_\omega(y))_x$  such that $|A_\omega(x)|\leq 1$, $|B_\omega(y)|\leq 1$ for every $x$, $y$, $\omega$; and the second supremum runs over all possible numbers $t_x=\pm1$, $s_y=\pm1$. The equality between both suprema follows from convexity.

We also have 
\begin{align}\label{quantum bias XOR}
\omega^*(G)=\sup\Big|\sum_{x\in \setx, y\in\sety}T_{x,y}tr(A_x\otimes B_y \rho)\Big|=\sup \Big|\sum_{x\in \setx,y\in \sety}T_{x,y}\langle u_x, v_y\rangle\Big|.
\end{align}
Here, the first supremum runs over all Hilbert spaces $H_1$ and $H_2$, all density operators $\rho$ acting on $H_1\otimes H_2$ and all families of self-adjoint norm-one operators $(A_{x})_{x}$, $(B_{y})_{y}$ acting on $H_1$ and $H_2$ respectively. The second equality is a well known result due to Tsirelson \cite{Tsirelson} and the supremum is taken over all families of vectors $(u_x)_x$, $(v_y)_y$ in a unit ball of a real Hilbert space. While there are many known XOR games $G$ for which $\omega^*(G)>\omega(G)$, Tsirelson's description of $\omega^*(G)$ allows us to use Grothendieck's inequality to state that, for any XOR game $G$,
\begin{align}\label{Gro-ineq}
1\leq \frac{\omega^*(G)}{\omega (G)}\leq K_G^{\R}\leq  1.7822\cdots .
\end{align}
\subsection*{One-way classical communication}

In this section we describe the one-way classical communication value of a XOR game when both players have unlimited classical resources (that is, they share an unlimited amount of randomness) and, additionally, Alice is allowed to send $c$ classical bits to Bob.

Let us assume that Alice and Bob receive inputs $x$ and $y$ respectively according to the probability distribution $\pi$. Then, Alice's answer  and message can depend only on the input $x$ and the randomness, which will be modeled via a probability space $(\Lambda, \lambda)$. Therefore, it can be modelled by a function $\theta:\setx \times \Lambda \longrightarrow \{-1,1\}\times [2^c]$ so that $\theta(x,\lambda)=(a,m)$, is the pair formed by Alice's answer $a$ and message $m$ when receiving input $x$ with shared randomness $\lambda$. We can consider the first and second components of $\theta$,  $a(x,\lambda)$ and $m(x,\lambda)$. At the same time, Bob's answer can only depend on the input $y$, the randomness $\lambda$ and the message $m$ received from Alice. Therefore, it can be modelled by a function $b: \sety\times \Lambda \times [2^c]\longrightarrow \{-1,1\}.$

That is, their joint correlation can be described by
$$\gamma(x,y)=\int_{\Lambda} a(x,\lambda) b(y, \lambda, m(x, \lambda)) d\lambda.$$

We will remark here that we could defined analogously the probability distribution $P=(P(a,b|x,y))_{a, b, x, y}$. However, since we will restrict the study of the quantities $\omega_{o.w-c}$ and  $\omega^*_{o.w-c}$ to the case of XOR games, it suffices to describe the correlation matrices.

Note also that, for every fixed $\lambda$, the correlation we obtain can be written as \begin{equation}\label{extremecorrelation} \gamma_\lambda(x,y)=a(x) b(y,  m(x)).\end{equation}

One can deduce easily from here that these are the extremal points of the set of the possible correlations. Now, we can easily prove the next result. 

\begin{lemma}\label{Lemma: classical}
Let $G=(\pi, f)$ be a XOR game with coefficients  $T_{x,y}=\pi(x,y) f(x,y)$ for every $x,y$. Then, 
\begin{align}\label{eq: classical-com}
\omega_{o.w-c}(G)=  \sup \Big|\sum_{x\in \setx, y\in\sety}\sum_{m=1}^{2^c} T_{x,y}a(x,m)b(y,m)\Big|.
\end{align}Here, the supremum runs over all families of real numbers $(a(x,m))_{x,m}$, $(b(y,m))_{y,m}$ verifying $\sum_{m=1}^{2^c} |a(x,m)|\leq 1$ for every $x$ and $|b(y,m)|\leq 1$ for every $y,m$.
\end{lemma}
\begin{proof}
According to (\ref{bias}) we only need to consider the correlations of our strategies. Then, it follows from convexity that we only need to maximize on the strategies as in Eq. (\ref{extremecorrelation}). That is,
\begin{align}\label{classical simplified}
\omega_{o.w-c}(G)=  \sup \Big|\sum_{x\in \setx, y\in\sety}\sum_{m=1}^{2^c} T_{x,y}\tilde{a}(x)\delta_{m,m(x)}\tilde{b}(y,m)\Big|,
\end{align}where the supremum runs over all families of real sequences $(\tilde{a}(x))_x$, $(\tilde{b}(y,m))_{y,m}$ and functions $m:\setx\rightarrow [2^c]$ with $\tilde{a}(x)=\pm 1$ and $\tilde{b}(y,m)=\pm 1$ for every $x$, $y$ ,$m$. Here, $\delta_{m,m(x)}$ equals one if $m=m(x)$ and equals zero otherwise.

If we define $\tilde{a}(x,m)=\tilde{a}(x)\delta_{m,m(x)}$ for every $x$ and $m$, it is straightforward to check that the supremum in (\ref{eq: classical-com}) is an upper bound of $\omega_{o.w-c}(G)$. In order to show that the suprema in  (\ref{eq: classical-com}) and (\ref{classical simplified}) are the same we will see that the elements $(\tilde{a}(x,m))_{x,m}$, $(\tilde{b}(y,m))_{y,m}$ are actually extreme points of the set where the supremum in (\ref{eq: classical-com}) is taken.

On the one hand, note that any vector $(b(y,m))_{y,m}$ with $|b(y,m)|\leq 1$ for every $y,m$ can be written as as a convex combination of vectors whose coordinates are $\pm 1$; that is, vectors like $(\tilde{b}(y,m))_{y,m}$. Indeed,  this means that the set of extreme points of a hypercube consists of its vertices. On the other hand, it is easy to see that any vector $(a(x,m))_{x,m}$ verifying $\sum_{m=1}^{2^c} |a(x,m)|\leq 1$ for every $x$ can be written as a convex combination of vectors of the form $(c(x,m))_{x,m}$ such that for every $x$ there exists an $m(x)$ such that $c(x,m(x))=\pm 1$ and $c(x,m)=0$ for $m\neq m(x)$. These are precisely vectors like $(\tilde{a}(x,m))_{x,m}$.
\end{proof}
\begin{remark}\label{r:classicalnorm}
Note that the expression for the quantity $\omega_{o.w-c}(G)$ appearing in  Lemma (\ref{Lemma: classical}) is not the simplest one. Indeed, if one restricts the optimization to the extreme points one can optimize as in (\ref{classical simplified}). However, writing $\omega_{o.w-c}(G)$ in its most general form will be useful to understand it as a norm of certain tensor. We will discuss this point in Section \ref{S:tensor norms}.
\end{remark}
\subsection*{One-way quantum communication}

Let us now describe the one-way quantum communication value of the game. In this model, the players have unlimited computational resources and shared randomness. Also, they are assisted with $c$ qubits which can be sent from Alice to Bob, but they are not allowed to share additional entanglement. Here, $S^{2^c}$ denotes the space of the $2^c$-dimensional quantum states.

Again, Alice and Bob receive inputs $x$ and $y$ respectively according to the probability distribution $\pi$. In this situation, Alice's answer  and message can be modelled by a function $\theta:\setx \times \Lambda \longrightarrow \{-1, 1\}\times [S^{2^c}]$ so that $\theta(x,\lambda)=(a,\rho)$, is the pair formed by Alice's answer $a$ and  a $2^c$ dimensional  quantum state  $\rho$ when receiving input $x$ with shared randomness $\lambda$. We can consider the first and second components of $\theta$,  $a(x,\lambda)$ and $\rho(x,\lambda)$. Note that $\rho(x,\lambda)$ could be entangled with a state on Alice's side, but whatever measurements she does on her side can be considered to be done prior to the sending of $\rho(x,\lambda)$. Bob's answer can be modelled by the result of a $\pm 1$ valued measurement he does on the quantum state he receives. The measurement he uses can  depend on the input $y$ and  the randomness $\lambda$.   Therefore, we can consider a function $B: \sety\times \Lambda \longrightarrow M_{2^c}^{s.a},$
where $M_{2^c}^{s.a}$ is the space of self-adjoint $2^c\times 2^c$ complex matrices (endowed with the operator norm) on Bob's side and, for every $(y,\lambda)$, $B(y,\lambda)=F_{y,\lambda}^{1}- F_{y,\lambda}^{1-}$, where $\{F_{y,\lambda}^{1}, F_{y,\lambda}^{1-}\}$ is a dichotomic POVM taking values $\pm 1$.  Note that this is equivalent to require that $B(y,\lambda)$ is a self-adjoint operator verifying $\|B(y,\lambda)\|_{M_{2^c}}\leq 1.$ After receiving $y$ and $\rho$, Bob will measure $\rho$ with the POVM and will answer $b=\pm 1$ depending on the result he obtains.

Hence, any correlation can be written as $$ \gamma(x,y)= \int_\Lambda a(x,\lambda)tr\big(B(y,\lambda) \rho(x,\lambda)\big) d\lambda,$$ so that they are all convex combinations of correlations of the form
\begin{equation}\label{qg1}
\gamma(x,y)=a(x)tr\big(B(y) \rho(x)\big).
\end{equation}
For the next lemma, given any $2^c\times 2^c$ complex matrix $A$ we use $\|A \|_{S_1^{2^c}}$ to denote the trace class norm of $A$. 
\begin{lemma}\label{l:quantumvalue}
Let $G=(\pi, f)$ be a XOR game with coefficients  $T_{x,y}=\pi(x,y) f(x,y)$ for every $x,y$. Then, 
\begin{equation}\label{qg}
\omega^*_{o.w-c}(G)= \sup \Big|\sum_{x\in \setx,y\in \sety}T_{x,y}tr\big(B_y  R_x\big)\Big|.
\end{equation}Here, the supremum runs over all families of self-adjoint operators $(B_{y})_{y}$, $(R_x)_{x}$ in $M_{2^c}$ verifying $\|R_x\|_{S_1^{2^c}}\leq 1$ for every $x$ and $\|B_y\|_{M_{2^c}}\leq 1$.
\end{lemma}
\begin{proof}
It suffices to consider correlation matrices and, in addition, to restrict to extreme points.
Hence, 
\begin{align}\label{eq: quantum-com}
\omega^*_{o.w-c}(G)= \sup \Big|\sum_{x\in \setx,y\in \sety}T_{x,y}a(x)tr\big(B(y) \rho(x)\big)\Big|,
\end{align}where $a(x)=\pm 1$, $\rho_x\in S^{2^c}$ for every $x$, and $(B_{y})_{y}$ is a family of self-adjoint operators verifying $\|B_y\|_{M_{2^c}}\leq 1$ for every $y$. 

Now, if we define $R_x=a(x) \rho(x)$, it is clear that $\omega^*_{o.w-c}(G)$ is upper bounded by the supremum in (\ref{qg}). In order to show the equality between the suprema in (\ref{qg}) and (\ref{eq: quantum-com}) we will see that the extreme points of the corresponding sets are the same. 

To this end, let us consider a family of self-adjoint  operators $(R_x)_{x}$ in $M_{2^c}$ verifying $\|R_x\|_{S_1^{2^c}}\leq 1$ for every $x$. It is easy to see that $(R_x)_{x}$  can be written as a convex combination of families of self-adjoint  operators $(\tilde{R}_x)_{x}$ in $M_{2^c}$ verifying $\|\tilde{R}_x\|_{S_1^{2^c}}= 1$ for every $x$. Indeed, this can be seen by considering the singular value decomposition of any $R_{x}$ and writing  the corresponding diagonal matrix $D_x$ as a suitable convex combination of diagonal matrices with only one non zero entry equal $\pm 1$.
 
On the other hand, if we consider a family of self-adjoint operators verifying $\|R_x\|_{S_1^{2^c}}= 1$, for every $x$, we can decompose  each operator $R_x$ in its positive and negative part $R_x=R_x^+-R_x^-$ (so that $tr(R_x^+)+tr(R_x^-)=1$) and write $$R_x=tr(R_x^+)\frac{R_x^+}{tr(R_x^+)}-tr(R_x^-)\frac{R_x^-}{tr(R_x^-)}.$$

It is clear that the operators $\rho^{\pm}_x=R_x^\pm/tr(R_x^\pm)$ are states. In addition, it is easy to define from here a probability space  $(\Lambda, \lambda)$, a family of $\pm 1$ random variables $(\alpha(x,\lambda))_{\lambda, x}$ and a family of states $(\rho(x,\lambda))_{\lambda, x}$ such that for fixed $(B_y)_y$, and for every $x\in \setx$,  we have $$tr(B_yR_x)=\int_{\Lambda} \alpha(x, \lambda)tr(B_y\rho(x,\lambda))d\lambda.$$



This concludes the proof.
\end{proof}

\begin{remark}\label{remark norm 1}
It follows from convexity reasonings that 
\begin{equation*}
\omega^*_{o.w-c}(G)= \sup \Big|\sum_{x\in \setx,y\in \sety}T_{x,y}tr\big(B_y  R_x\big)\Big|,
\end{equation*}where here the supremum runs over all families of self-adjoint operators $(B_{y})_{y}$, $(R_x)_{x}$ in $M_{2^c}$ verifying $\|R_x\|_{S_1^{2^c}}=1$ and $\|B_y\|_{M_{2^c}}= 1$ for every $x$, $y$. In fact, a compactness argument shows that the supremum is actually a maximum.
\end{remark}

Let us conclude this section by providing an upper bound for the quotient $\omega^*_{o.w-\log d}(G)/\omega_{o.w-\log d}(G)$.
\begin{prop}\label{prop optimality}
Let $G$ be a XOR game. Then, for every natural number $n$ the following inequalities hold:
\begin{align}\label{eq: upper bound}
\omega^*_{o.w-\log n}(G)\leq \sqrt{n}\cdot \omega^*(G)\leq K_G^{\mathbb R}\sqrt{n}\cdot \omega (G).
\end{align}
\end{prop}
\begin{proof}
The second inequality in (\ref{eq: upper bound}) is a consequence of (\ref{Gro-ineq}). 

Hence, it suffices to show the inequality $\omega^*_{o.w-\log n}(G)\leq \sqrt{n}\cdot \omega^*(G)$. To this end, let us consider the coefficients associated to the game $G$, $(T_{x,y})_{x,y}$. According to Lemma \ref{l:quantumvalue} and Eq. (\ref{quantum bias XOR}), we must prove that
\begin{equation*}
\sup \Big|\sum_{x\in \setx,y\in \sety}T_{x,y}tr\big(R_xB_y\big)\Big|\leq K_G^{\mathbb R}\sqrt{n}\sup \Big|\sum_{x\in \setx,y\in \sety}T_{x,y}\langle u_x, v_y\rangle\Big|, 
\end{equation*} where the supremum on the left hand side runs over all families of self-adjoint operators $(R_x)_{x}$, $(B_{y})_{y}$ verifying $\|R_x\|_{S_1^{n}}\leq 1$ and $\|B_y\|_{M_{n}}\leq 1$ for every $x$, $y$; and the supremum on the right hand side is taken over all families of vectors $(u_x)_x$, $(v_y)_y$ in a unit ball of a real Hilbert space. 

Now, for one such family of self-adjoint operators $(R_x)_{x}$,  $(B_{y})_{y}$, we know that $\|R_x\|_{S_2^{n}}\leq 1$ and $\|B_y\|_{S_2^{n}}\leq \sqrt{n}$. Indeed, these inequalities follow from the well known facts $\|\cdot \|_{S_2^{n}}\leq \|\cdot \|_{S_1^{n}}$ and $\|\cdot \|_{M_n}\leq \sqrt{n}\|\cdot \|_{S_2^{n}}$ (which can be easily checked by considering the singular value decomposition of the matrices). Then, we can realize $(R_x)_x$, $(\frac{1}{\sqrt{n}}B_y)_y$ as two families of elements in the unit ball of the real Hilbert space of self-adjoints operators in $M_{n}$ endowed with the inner product $\langle A, B\rangle=tr(AB)$. That is, we can see 
$$tr\big(B_y  R_x\big)=\sqrt{n}\langle u_x, v_y\rangle$$for some vectors $(u_x)_x$, $(v_y)_y$ in the unit ball of a real Hilbert space. This concludes the proof.
\end{proof}
\section{Bell violations from advantages in the value of the game with communication}\label{S:Bell violations}
In this section we prove the first of our main results, Theorem \ref{Thm:MainI}. For the convenience of the reader let us recall the statement of the theorem. Given a XOR game $G=(\pi,f)$ with coefficients $T_{x,y}=\pi(x,y)f(x,y)$ for every $x\in \setx$, $y\in \sety$, and given a natural number $d$, we will consider a Bell functional $B^G_d$ defined as follows:

- Set of inputs for Alice and Bob: $\tilde{\setx}=\setx$ and $\tilde{\sety}=\sety\times \{1,\cdots, d\}$, respectively.

- Set of outputs for Alice and Bob: $\tilde{\seta}=\{1,\cdots, d\}\times \{-1,1\}$ and $\tilde{\setb}=\{-1,1\}$, respectively.

The coefficients of the Bell functional are defined, for every $x\in \setx$, $(y,k)\in \sety\times \{1,\cdots, d\}$, $(a,\tilde{a})\in \{1,\cdots, d\}\times \{-1,1\}$, $b\in \{-1,1\}$, as:
\begin{align*}
B^G_d(a,\tilde{a},b,x,y,k)=T_{x,y}\cdot \delta_{a,k} \cdot \tilde{a}\cdot b,
\end{align*}where $\delta_{a,k}$ equals one if $a=k$ and equals zero otherwise.

Then, Theorem \ref{Thm:MainI} states that 
$$\frac{\omega^*(B^G_{d^2})}{\omega(B^G_{d^2})}\geq \frac{\omega^*_{o.w-\log d}(G)}{\omega_{o.w-2\log d}(G)}.$$

Theorem \ref{Thm:MainI} will be obtained as a direct consequence of the following two lemmas, which relate the values $\omega_{o.w-2\log d}(G)$ and $\omega^*_{o.w-\log d}(G)$ with the classical and quantum values of the Bell functional $B^G_{d^2}$ respectively.
\begin{lemma}\label{lemma 1}
Let $G$ be a XOR and let $d$ be a natural number. Let $B^G_{d^2}$ be the Bell functional defined above. Then, 
$$\omega(B^G_{d^2})\leq \omega_{o.w-2\log d}(G).$$
\end{lemma}
\begin{proof}
According to the definition of the classical value of a Bell functional (\ref{classical value}), by convexity we just need to look at elements of the form $P(x, y,k|a,\tilde{a},b)=P(a, \tilde{a}|x)Q(b|y,k)$ for every $x$, $(y,k)$, $(a, \tilde{a})$, $b$; where $(P(a, \tilde{a}|x))_{a,\tilde{a},x}$, $(Q(b|y,k))_{b,y,k}$ are nonnegative numbers verifying $\sum_{a=1}^{d^2}\sum_{\tilde{a}=\pm 1}P(a, \tilde{a}|x)=1$ and  $Q(1|y,k)+Q(-1|y,k)=1$ for every $x$, $y$ and $k$. Then,

\begin{align*}
\langle B^G_{d^2}, P\rangle&=\sum_{x\in \setx, y\in \sety}\sum_{k=1}^{d^2}\sum_{a=1}^{d^2}\sum_{\tilde{a}=\pm 1}\sum_{b=\pm 1} T_{x,y}\cdot \delta_{a,k}\cdot \tilde{a}\cdot b \cdot P(a, \tilde{a}|x) Q(b|y,k)\\
&=\sum_{x\in \setx, y\in \sety}\sum_{k=1}^{d^2} T_{x,y}\big(P(k, 1|x)-P(k,-1|x)\big)\big(Q(1|y,k)-Q(-1|y,k)\big)\\
&=\sum_{x\in \setx, y\in \sety}\sum_{k=1}^{d^2} T_{x,y}\gamma(k,x) \beta(y,k),
\end{align*}where $\gamma(k,x)=P(k, 1|x)-P(k,-1|x)$ and $\beta(y,k)=Q(1|y,k)-Q(-1|y,k)$ for every $x,y,k$. Note that these are real numbers verifying $|\beta(y,k)|\leq 1$ for every $y,k$, and $$\sum_{k=1}^{d^2}|\gamma(x,k)|=\sum_{k=1}^{d^2}|P(k, 1|x)-P(k,-1|x)|\leq \sum_{k=1}^{d^2}\Big(P(k, 1|x)+P(k,-1|x) \Big)= 1$$for every $x$.

According to Lemma \ref{Lemma: classical}, $$\omega(B^G_{d^2})\leq\omega_{o.w-\log d^2}(G)=\omega_{o.w-2\log d}(G).$$
\end{proof}

In order to study the quantum case, we will consider, for every natural number $n$, the following unitaries on $\C^n$ defined as $$u_j |l\rangle=e^{\frac{2\pi ijl}{n}}|l\rangle\text{    }\text{    }\text{ and    }\text{    }\text{    }v_k |l\rangle=|l+k \rangle\text{    }\text{    }\text{   for every  }j,k,l=1,\cdots, n,$$where $j+l$ is understood mod $n$. Then, we consider the new unitaries $$W_{k,j}=v_ku_j \text{    }\text{   for every  }j,k=1,\cdots, n.$$

An important property of these unitaries is that 
\begin{align}\label{unitary properties}
\frac{1}{n}\sum_{j,k=1}^nW_{k,j}A W_{k,j}^*=tr(A)\uno_{M_n}
\end{align}for every matrix $A$ in $M_n$.

The previous unitaries have been used in different context of quantum information. In fact, the proof of the following result is motivated by the embeddings between noncommutative $L_p$-spaces considered in \cite{JPII}, which are themselves based on the quantum teleportation protocol. However, no knowledge about noncommutative $L_p$-spaces will be needed here.

\begin{lemma}\label{lemma 2}
Let $G$ be a XOR and let $d$ be a natural number. Let $B^G_{d^2}$ be the Bell functional defined above. Then, 
$$\omega^*(B^G_{d^2})\geq  \omega^*_{o.w-\log d}(G).$$
\end{lemma}
\begin{proof}
According to Lemma \ref{l:quantumvalue} and Remark \ref{remark norm 1} there exist families of self-adjoint operators $(R_x)_{x}$, $(B_{y})_{y}$ verifying $\|R_x\|_{S_1^d}=  1$, $\|B_y\|_{M_d}= 1$ for every $x$, $y$; and such that $$\omega^*_{o.w-\log d}(G)= \Big|\sum_{x,y}T_{x,y}tr\big(R_xB_y  \big)\Big|.$$ 

In particular, for every $y$ we can write $B_y=F_y^1-F_y^{-1}$ for certain semidefinite positive operators $F_y^i$, $i=-1,1$ verifying $F_y^1+F_y^{-1}=\uno$. At the same time, for every $x$ we can write $R_x=R_x^{1}-R_x^{-1}$ for some semidefinite positive operators $R_x^{i}$, $i=1,2$ verifying $1=\|R_x\|_{S_1^d}=tr(R_x^{1})+tr(R_x^{-1})$.

Let us identify the sets $\{1,\cdots, d^2\}=\{(j,k): j,k=1,\cdots , d\},$ so that we can talk about $W_a$ for every $a=1,\cdots, d^2$.

Then, for every $x\in \setx$ and $a=1,\cdots, d^2$ we define $$E_x^{a,1}=\frac{1}{d}W_{a}R_x^{1}W_{a}^*, \text{     }\text{     }\text{  and   } \text{     }\text{     } E_x^{a,-1}=\frac{1}{d}W_{a}R_x^{-1}W_{a}^*.$$These operators are semidefinite positive and, according to (\ref{unitary properties}), for every $x$ we have $$\sum_{a=1}^{d^2}\sum_{\tilde{a}=\pm 1}E_x^{a, \tilde{a}}=\frac{1}{d}\sum_{a=1}^{d^2}W_{a}(R_x^{1}+R_x^{-1})W_{a}^*=tr(R_x^{1}+R_x^{-1})\uno_{M_d}= \uno_{M_d}.$$

Now, for every $y\in \sety$, $k=1,\cdots , d^2$ and $i=1,-1$, we define $P^i_{y,k}=\bar{W}_k(F_y^i)^TW_k^T$, where $\bar{W}_k$ and $W_k^T$ are the conjugate operator and the transpose operator of $W_k$ respectively. The operators $P^i_{y,k}$ are semidefinite positive and they verify $P^1_{y,k}+P^{-1}_{y,k}=\uno$ for every $y,k$. 

The set of operators $(E_x^{a, \tilde{a}})_{a, \tilde{a}, x}$ and $\{P^b_{y,k}\}_{i,y,k}$ define families of POVMs acting on $\C^d$ for Alice and Bob respectively. Let us assume that Alice and Bob share the maximally entangled state $|\varphi\rangle=\frac{1}{\sqrt{d}}\sum_{i=1}^d|ii\rangle$. Then,
\begin{align*}
\omega^*(B^G_{d^2})&\geq\Big| \sum_{x\in \setx, y\in \sety}\sum_{k=1}^{d^2}\sum_{a=1}^{d^2}\sum_{\tilde{a}=\pm 1}\sum_{b=\pm 1}  T_{x,y}\cdot \delta_{a,k}\cdot \tilde{a}\cdot b \cdot \langle \varphi |E_x^{a, \tilde{a}}\otimes P^b_{y,k}|\varphi \rangle\Big|\\&=\frac{1}{d^2}\Big|\sum_{x\in \setx, y\in \sety}\sum_{k=1}^{d^2} T_{x,y}tr\Big (W_{k}\big (R_x^{1}-R_x^{-1}\big )W_{k}^* W_k(F_y^1- F_y^{-1})W_k^*\Big)\Big|\\&=\frac{1}{d^2}\Big|\sum_{x\in \setx, y\in \sety}\sum_{k=1}^{d^2}  T_{x,y}tr\Big (R_xB_y\Big)\Big|=\\&=\Big|\sum_{x\in \setx, y\in \sety}T_{x,y}tr\Big (R_xB_y\Big)\Big|\\&= \omega^*_{o.w-\log d}(G).
\end{align*}
\end{proof}

As we said before, Theorem \ref{Thm:MainI} follows immediately from Lemma \ref{lemma 1} and Lemma \ref{lemma 2}.
\section{Bell violations from  communication complexity advantages}\label{S:BV}

We start this section showing  a relation between the quotient $\omega^*_{o.w.-c}(G)/\omega_{o.w.-2c}(G)$ for XOR games $G=(f,\pi)$ and the advantage in quantum versus classical communication complexity of the function $f$. This relation will allows us to obtain, in Theorem \ref{Thm:MainII}, Bell inequality violations starting from a constant ratio advantage in the communication complexity of a function $f$. Note that the main result in \cite{Buhrman16} is the derivation of Bell inequality violations from advantage in communication complexity, but in that paper it is required that $CC(f,\frac{2}{3})\geq QC(f,\frac{2}{3})^4$, whereas we only require $CC(f,\frac{2}{3})\geq \alpha QC(f,\frac{2}{3})$ for constant $\alpha$. 

The following lemma uses standard techniques.  A proof appears in \cite[Appendix A]{Buhrman16}.

\begin{lemma}\label{l:highcomplexity}
Let $0<\epsilon<\frac{1}{6}$ and $f:\setx\times \sety\longrightarrow \{-1,1\}$. Then  $$CC(f, \frac{1}{2} +\epsilon)\geq \frac{\epsilon^2}{3} CC(f, \frac{2}{3}).$$
\end{lemma}

Before we state and prove our next result, we need a couple of observations. 

First, we recall the following definition (see, for instance, \cite[Definition 3.19]{NiKu}). For a function $f:\setx\times \sety\longrightarrow \{-1,1\}$, and a probability distribution $\pi:\setx\times \sety\longrightarrow [0,1]$, the (one-way) $(\pi,\delta)$-{\em distributional complexity} of $f$, $D_\pi(f,\delta)$ is the cost of the best (one-way) deterministic protocol that gives the correct answer for $f$ on at least a $\pi$ fraction of all inputs in $\setx\times \sety$, weighted by $\mu$. 

Next, we observe that if we have a XOR game $G=(f,\pi)$ for which we have a strategy which allows us to guess the correct answer with probability $\frac{1}{2}+\epsilon$, then, according to (\ref{bias}), for that strategy we have  $$\omega(G)=(\frac{1}{2}+\epsilon)-(\frac{1}{2}-\epsilon)=2\epsilon.$$

We recall again that our quantity $\omega(G)$ corresponds to the bias of the XOR games $G$ and not to the value of it.

Now we can state and prove the following lemma.

\begin{lemma}\label{lemma Thm II}
Let $\alpha>216$ and let $f:\setx\times \sety \longrightarrow \{-1, 1\}$ be a function for which $$\frac{CC(f,\frac{2}{3})}{QC(f,\frac{2}{3})}> \alpha.$$ 

Then, calling $c=QC(f,\frac{2}{3})$, there exists a probability distribution $\pi: \setx\times \sety \longrightarrow [0,1]$ such that the XOR game $G=(f, \pi)$ verifies $$\frac{\omega^*_{o.w.-c}(G)}{\omega_{o.w.-2c}(G)}> \frac{\sqrt{\alpha}}{6\sqrt{6}}.$$
\end{lemma}

\begin{proof}
Since $CC(f, \frac{2}{3})\geq \alpha QC(f, \frac{2}{3})$, Lemma \ref{l:highcomplexity} implies that for every $0<\epsilon <\frac{1}{6}$, $$CC(f, \frac{1}{2}+\epsilon ) > \frac{\alpha QC(f, \frac{2}{3})\epsilon^2}{3}.$$

Choosing $\epsilon=\sqrt{\frac{6}{\alpha}}$, we get  $CC(f, \frac{1}{2}+ \epsilon )> 2QC(f, \frac{2}{3})$. Note that the fact $\alpha>216$ guarantees that $0<\epsilon<\frac{1}{6}$.

We apply now  \cite[Theorem 3.20]{NiKu}, and we have that there exists a probability distribution $\pi:\setx\times \sety\longrightarrow [0,1]$ such that  $D_\pi(f, \frac{1}{2}+\epsilon) > 2QC(f, \frac{2}{3})$, where $D_\pi(f, \frac{1}{2}+\epsilon) $ denotes the one-way distributional complexity  of $f$ with the probabiity distribution $\pi$ (Note that \cite[Theorem 3.20]{NiKu} is stated and proved for two-way communication complexity, but the same result, with the same proof, applies for the case of one-way communication complexity). Calling $c=QC(f,\frac{2}{3})$  and $G=(f, \pi)$, this  implies that ${\omega_{o.w.-2c}(G)}\leq 2\epsilon=2\sqrt{\frac{6}{\alpha}}$. At the same time, the fact that $c=QC(f,\frac{2}{3})$ implies that, for every choice of a probability distribution $\nu:\setx\times \sety\longrightarrow [0,1]$, in particular for $\nu=\pi$, the game $G=(f, \nu)$ verifies $\omega^*_{o.w.-c}(G)\geq 2(\frac{2}{3}-\frac{1}{2})=\frac{1}{3}$. The result now follows. 
\end{proof}

Theorem \ref{Thm:MainII} is now straightforward from Theorem \ref{Thm:MainI} and Lemma \ref{lemma Thm II}.

To compare this result with the main result of \cite{Buhrman16}, note that in that paper, to achieve a Bell inequality violation it is required that  $CC(f,\frac{2}{3})\geq QC(f,\frac{2}{3})^4$, whereas we only need constant separation. In addition, let us mention that condition $\alpha>216$ comes from the application of Lemma \ref{l:highcomplexity} and the definition of $\epsilon=\sqrt{\frac{6}{\alpha}}$ in the proof of Lemma \ref{lemma Thm II}. Since our goal was to show that  Theorem \ref{Thm:MainI} allows to obtain Bell inequality violations from a constant separation in communication complexity problems we did not make a great effort to improve condition $\alpha>216$. However, we think that an alternative proof of Lemma \ref{lemma Thm II} should allow to decrease the value $216$. 

Let us also remark that, whereas Theorem \ref{Thm:MainII} applies only to the one-way communication scenario, the main result in \cite{Buhrman16} covers the case of two-way communication complexity scenarios. This is in fact one of the main points in \cite{Buhrman16}. We still do not know if our techniques can be adapted to deal with that situation.

\section{A game almost maximizing the ratio $\omega^*_{o.w.-c}(G)/\omega_{o.w.-c}(G)$}\label{S:values}

In this section we define a XOR game $G$  for which the  ratio $\omega^*_{o.w-c}(G)/ \omega_{o.w-c}(G)$  is essentially optimal as a function of $c$. Indeed, we will prove Theorem \ref{Thm:MainIII}, which could be stated in the following more precise form:
\begin{theorem}
There exists a family of XOR games $(G_n)_n$ so that $G_n$ has $2^{2n}$ inputs for Alice and $2^{n^2}$ inputs for Bob and such that for every $k\geq e^2$ the following inequality holds:$$\frac{\omega^*_{o.w-\log n}(G_n)}{\omega_{o.w-\log k}(G_n)} \geq C\frac{\sqrt{n}}{\log k},$$where $C$ is a universal constant.

Moreover, this statement is essentially optimal, up to a logarithmic factor. 
\end{theorem}

Note that the optimality follows from Proposition \ref{prop optimality}. So we must prove the existence of such a family of games. For the sake of simplicity we will remove the dependence on $n$ and we will just write $G$. 

In our game, the set of inputs for Alice and Bob are respectively $\tilde{\setx}=\setx\times \setz =\{-1,1\}^{n}\times \{-1,1\}^{n}=\{-1,1\}^{2n} $ and $\sety=\{-1,1\}^{n^2}$. In the following, we will often write $\tilde{x}=(x,z)\in \tilde{\setx}=\setx\times \setz$. To define the probability distribution on the the set of inputs, first we define the number 
\begin{align}\label{Def M}
M=\sum_{\tilde{x}\in \tilde{\setx}, y\in \sety}\left |\sum_{i,j=1}^n x_iz_jy_{i,j}\right|.
\end{align}

Then, the probability distribution $\pi$ on the set of inputs and the function $f$ will be given, for every $\tilde{x}\in \tilde{\setx}$ and $y\in \sety$, by
\begin{align}\label{pi-f Game}
\pi(\tilde{x},y)=\frac{1}{M}\left|\sum_{i,j=1}^nx_iz_jy_{i,j}\right|, \text{    }\text{    }\text{  and   }\text{    }\text{    } f(\tilde{x},y)=\sign\left(\sum_{i,j=1}^nx_iz_jy_{i,j}\right).
\end{align}

In particular, the coefficients of our game are
\begin{align}\label{coefficients of G}
T_{\tilde{x},y}=\pi(\tilde{x},y)f(\tilde{x},y)=\frac{1}{M}\sum_{i,j=1}^nx_iz_jy_{i,j}
\end{align}for every $\tilde{x}=(x,z)\in  \{-1,1\}^{n}\times \{-1,1\}^{n}$ and $y\in \{-1,1\}^{n^2}$.

Given any natural number $n$ and $1\leq i\leq n$, we denote the\emph{ Rademacher function} $r_i:\{-1,1\}^n\rightarrow \{-1,1\}$ by $$r_i(w)=w(i) \text{        }\text{   for any     }\text{        }w\in \{-1,1\}^n.$$

Let us also denote, for any $1\leq p<\infty$, by $\ell_p^{2^n}$ the space of all functions $f:\{-1,1\}^n\rightarrow \R$ with the norm
\begin{align}\label{Def p-norm}
\|f\|_p=\Big(\sum_{w\in \{-1,1\}^n} |f(w)|^p\Big)^\frac{1}{p}.
\end{align}

It is well known that, with the previous notation, we have the duality relation
\begin{align}\label{duality relation}
\|f\|_p=\sup \{|\langle f, g\rangle|: \|g\|_{p'}\leq 1\}
\end{align} 
for every $1<p, p'<\infty$   such that    $\frac{1}{p}+\frac{1}{p'}=1$, where we denote $$\langle f, g\rangle=\sum_{w\in \{-1,1\}^n}f(w)g(w).$$

A key point in our analysis is the Khintchine inequality, that we state here. In fact, we will also use the double Khintchine inequality. The proof of these results can be found in \cite[pag. 96]{DF} and \cite[pag. 455]{DF} respectively.

\begin{theorem}\label{Khintchine ineq}
For $1\leq p <\infty$ there exist constants $a_p, \text{}b_p\geq 1$ such that
\begin{align}\label{standard Khintchine ineq}
a_p^{-1}\left(\sum_{i=1}^n |\alpha_i|^2\right)^\frac{1}{2}\leq \left(\sum_{w\in \{-1,1\}^n}\frac{1}{2^n}\Big|\sum_{i=1}^n \alpha_ir_i(w)\Big|^p\right)^{\frac{1}{p}} \leq b_p\left(\sum_{i=1}^n |\alpha_i|^2\right)^\frac{1}{2}
\end{align}for every $n$ and all $\alpha_1, \cdots, \alpha_n \in \mathbb C$.

Moreover, 
\begin{align}\label{double Khintchine ineq}
a_p^{-2}\left(\sum_{i,j=1}^n |\alpha_{i,j}|^2\right)^\frac{1}{2}\leq \left(\sum_{w,w'\in \{-1,1\}^n}\frac{1}{2^{2n}}\Big|\sum_{i,j=1}^n \alpha_{i,j}r_i(w)r_j(w')\Big|^p\right)^{\frac{1}{p}} \leq b_p^2\left(\sum_{i,j=1}^n |\alpha_{i,j}|^2\right)^\frac{1}{2}
\end{align}for every $n$ and all $\alpha_{1,1}, \alpha_{1,2}, \cdots, \alpha_{n,n} \in \mathbb C$.
\end{theorem}

Although we will not need to know the value of the constant $M$ defined above, it is not difficult to compute its value up to a constant.
\begin{lemma}\label{control of M}
Given $M$ defined as in (\ref{Def M}), then
\end{lemma}
\begin{align}\label{M value}
\frac{1}{\sqrt{2}}n2^{n^2+2n}\leq M\leq n2^{n^2+2n}.
\end{align}
\begin{proof}
For a fixed $(x,z)\in \{-1,1\}^{n}\times \{-1,1\}^{n}$, we have 
\begin{align*}
\sum_{y\in \{-1,1\}^{n^2}}\left|\sum_{i,j=1}^n x_iz_jy_{i,j}\right|=2^{n^2}\sum_{y\in \{-1,1\}^{n^2}}\frac{1}{2^{n^2}}\left|\sum_{i,j=1}^n x_iz_jy_{i,j}\right|. 
\end{align*}

Now, since it is known that $a_1=\sqrt{2}$ and $b_1=1$ (\cite[pag. 96]{DF}), by recalling that $y_{i,j}=r_{i,j}(y)$ for every $y\in \{-1,1\}^{n^2}$, (\ref{Khintchine ineq}) states that $$\frac{1}{\sqrt{2}}\Big(\sum_{i,j=1}^n| x_iz_j|^2\Big)^{\frac{1}{2}}\leq \sum_{y\in \{-1,1\}^{n^2}}\frac{1}{2^{n^2}}\left|\sum_{i,j=1}^n x_iz_jy_{i,j}\right|\leq \left(\sum_{i,j=1}^n| x_iz_j|^2\right)^{\frac{1}{2}}.$$ Using that $\left(\sum_{i,j=1}^n| x_iz_j|^2\right)^{\frac{1}{2}}=n$ for every $(x,z)\in \{-1,1\}^{n}\times \{-1,1\}^{n}$, the estimate (\ref{M value}) follows easily.
\end{proof}

The following lemma, which will be crucial in the analysis of $\omega_{o.w-\log n}(G)$ for our game, is the transposed version of the double Khintchine inequality (\ref{double Khintchine ineq}).
\begin{lemma}\label{l:transpose Khintchine} 
Let $1<p<\infty$ and let $p'$ be such that $\frac{1}{p}+\frac{1}{p'}=1$. 
For every finite sequence of numbers $(\alpha(x,z))_{(x,z)\in \{-1, 1\}^n\times  \{-1,1\}^n}$,
\begin{align*}
\left(\sum_{i,j=1}^n \left( \sum_{(x,z)} r_i(x) r_j(z) \alpha(x,z)\right)^2\right)^\frac{1}{2}\leq b_{p'}^2  \left(2^{2n}\right)^\frac{1}{p'} \left(\sum_{(x,z)} |\alpha(x,z)|^p\right)^\frac{1}{p}.
\end{align*}Here, the sums in $(x,z)$ are over $\{-1,1\}^n\times  \{-1,1\}^n$ and $b_{p'}$ is the constant appearing in Theorem \ref{Khintchine ineq} for $p'$.
\end{lemma}
\begin{proof}
Let us consider $\varphi=\left ( \sum_{(x,z)} r_i(x) r_j(z) \alpha(x,z)\right)_{i,j=1}^n $ as an element of $\mathbb C^{n^2}$. Then, we have   
\begin{align*}\left(\sum_{i,j=1}^n \Big( \sum_{(x,z)} r_i(x) r_j(z) \alpha(x,z)\Big)^2\right)^\frac{1}{2}=\sup_{\beta\in \mathbb C^{n^2}} \frac{1}{\|\beta\|_2} \langle \varphi | \beta\rangle =\sup_{\beta\in \mathbb C^{n^2}} 
\frac{1}{\|\beta\|_2} \sum_{i,j=1}^n \sum_{x,z} r_i(x) r_j(z) \alpha(x,z) \beta_{i,j}. 
\end{align*}

We can view now $\left(\alpha(x,z)\right)_{x,z}$ and $\left(\sum_{i,j=1}^n  r_i(x) r_j(z) \beta_{i,j}\right)_{x,z}$ as elements in $\ell_p^{2^{2n}}$ and $\ell_{p'}^{2^{2n}}$ respectively and, by duality (\ref{duality relation}),
\begin{align*}\sup_{\beta\in \C^{n^2}} 
\frac{1}{\|\beta\|_2} \sum_{i,j=1}^n \sum_{x,z} r_i(x) r_j(z) \alpha(x,z) \beta_{i,j} &\leq  \sup_{\beta\in \C^{n^2}} 
\frac{1}{\|\beta\|_2} \Big\| \Big(\sum_{i,j=1}^n  r_i(x) r_j(z) \beta_{i,j}\Big)_{x,z} \Big\|_{p'} \Big\| \Big(\alpha(x,z)\Big)_{x,z}\Big\|_{p}\\&\leq 
 b_{p'}^2 (2^{2n})^\frac{1}{p'}  \left(\sum_{x,z} |\alpha(x,z)|^p\right)^\frac{1}{p},
 \end{align*}where the  last inequality follows from (\ref{double Khintchine ineq}).
 
This finishes the proof.
\end{proof}

For the proof of Theorem \ref{Thm:MainIII} we will use that, as a consequence of Holder's inequality, for every sequence of real numbers $(\alpha_i)_{i=1}^d$,
\begin{align}\label{1-p'}
\sum_{i=1}^d |\alpha_i|\leq d^{\frac{1}{p'}} \left(\sum_{i=1}^d |\alpha_i|^p\right)^\frac{1}{p} \text{   }\text{  if  }\text{   } 1<p<\infty.
\end{align}

The proof of Theorem \ref{Thm:MainIII} will trivially follow from Proposition \ref{upper bound XOR game} and Proposition \ref{lower bound XOR game} below, where we provide upper and lower bounds for $\omega_{o.w.-\log n}(G)$ and $\omega^*_{o.w.-\log n}(G)$ respectively.

\begin{prop}\label{upper bound XOR game}
Let $G$ be the XOR game defined via (\ref{pi-f Game}). Then, for a given $k\geq e^2$ we have 
\begin{align*}
\omega_{o.w.-\log k}(G)\leq \frac{2\sqrt{2}e^2}{n}\ln k.
\end{align*}
\end{prop}
\begin{proof}
Let $(a(\tilde{x},m))_{\tilde{x},m}$ and $(b(y,m))_{y,m}$ be real numbers as in Lemma \ref{Lemma: classical}, where here $\tilde{x}=(x,z)\in \{-1,1\}^{n}\times \{-1,1\}^{n}=\{-1,1\}^{2n} $, $\tilde{y}\in\{-1,1\}^{n^2}$ and $m=1,\cdots, k$. Note that, since $\sum_{m=1}^c |a(\tilde{x},m)|\leq 1$ for every $\tilde{x}$, we can easily conclude that
\begin{align}\label{p-n}
\Big(\sum_{\tilde{x}, m} |a(\tilde{x},m)|^p\Big)^{\frac{1}{p}}\leq \Big(\sum_{\tilde{x}} 1\Big)^{\frac{1}{p}}=(2^{2n})^{\frac{1}{p}}.
\end{align}

Moreover, since  $|b(y,m)|\leq $ for every $y$ and $m$, we have
\begin{align*}
\omega_{o.w.-\log k}(G)&\leq \Big|\sum_{\tilde{x},y} \sum_{m=1}^{k}  T_{\tilde{x},y}a(\tilde{x},m)b(y,m)\Big|\leq \sum_{y,m}\Big|\sum_{\tilde{x}}   T_{\tilde{x},y}a(\tilde{x},m)\Big|\\&\leq \big(2^{n^2}k\big)^{\frac{1}{p'}} \Big(\sum_{y,m}\Big|\sum_{\tilde{x}}   T_{\tilde{x},y}a(\tilde{x},m)\Big|^p\Big)^{\frac{1}{p}},
\end{align*}where in the last inequality we have used (\ref{1-p'}).

Now, by writing the precise value of the coefficients $T_{\tilde{x},y}$ (\ref{coefficients of G}), and using (\ref{standard Khintchine ineq}) on the variable $y$, we have
\begin{align*}
&\sum_{y}\Big|\sum_{\tilde{x}}   T_{\tilde{x},y}a(\tilde{x},m)\Big|^p=2^{n^2}\sum_{y}\frac{1}{2^{n^2}}\Big|\sum_{\tilde{x}}   T_{\tilde{x},y}a(\tilde{x},m)\Big|^p=\frac{2^{n^2}}{M^p}\sum_{y}\frac{1}{2^{n^2}}\Big|\sum_{i,j=1}^n\Big(\sum_{\tilde{x}}   x_iz_jy_{i,j}a(\tilde{x},m)\Big)\Big|^p\\&=\frac{2^{n^2}}{M^p}\sum_{y}\frac{1}{2^{n^2}}\Big|\sum_{i,j=1}^nr_{i,j}(y)\Big(\sum_{(x,z)}   r_i(x)r_j(z)a(x,z,m)\Big)\Big|^p\leq \frac{2^{n^2}}{M^p}b_p^p\Big(\sum_{i,j=1}^n\Big|\sum_{(x,z)}   r_i(x)r_j(z)a(x,z,m)\Big|^2\Big)^{\frac{p}{2}}.
\end{align*}

Therefore, we deduce that 
\begin{align*}
\omega_{o.w.-\log k}(G)&\leq  \frac{1}{M}\big(2^{n^2}k\big)^{\frac{1}{p'}}(2^{n^2})^{\frac{1}{p}}b_p \Big(\sum_{m}\Big[\sum_{i,j=1}^n\Big|\sum_{(x,z)}   r_i(x)r_j(z)a(x,z,m)\Big|^2\Big]^{\frac{p}{2}}\Big)^{\frac{1}{p}}\\&=
 \frac{1}{M}2^{n^2}k^{\frac{1}{p'}}b_p \Big(\sum_{m}\Big[\sum_{i,j=1}^n\Big|\sum_{(x,z)}   r_i(x)r_j(z)a(x,z,m)\Big|^2\Big]^{\frac{p}{2}}\Big)^{\frac{1}{p}}\\&\leq 
 \frac{1}{M}2^{n^2}k^{\frac{1}{p'}}b_p b_{p'}^2\big(2^{2n}\big)^{\frac{1}{p'}}\Big(\sum_{m}\Big(\sum_{(x,z)}   \big|a(x,z,m)\big|^p\Big)^{\frac{1}{p}},
\end{align*}where in the last inequality we have used Lemma \ref{l:transpose Khintchine}. 

Finally, according to (\ref{p-n}) we deduce that 
\begin{align*}
\omega_{o.w.-\log k}(G)&\leq  
 \frac{1}{M}2^{n^2}k^{\frac{1}{p'}}b_p b_{p'}^2\big(2^{2n}\big)^{\frac{1}{p'}}\big(2^{2n}\big)^{\frac{1}{p}}= \frac{1}{M}2^{n^2+2n}k^{\frac{1}{p'}}b_p b_{p'}^2.
\end{align*}

Now, if we consider the particular choice $p'=\ln k$, then $b_p=1$ (since $p< 2$) and $b_{p'}\leq \sqrt{2ep'}=\sqrt{2e\ln k}$ (see \cite[Section 8.5]{DF}). Hence, according to lemma \ref{control of M} we conclude
\begin{align*}
\omega_{o.w.-\log k}(G)&\leq  \frac{2\sqrt{2}e^2}{n}\ln k.
\end{align*}
\end{proof}
\begin{prop}\label{lower bound XOR game}
Let $G$ be the XOR game defined via (\ref{pi-f Game}). Then, 
\begin{align*}
\omega^*_{o.w.-\log n}(G)\geq  \frac{C}{\sqrt{n}},
\end{align*}where here $C$ is a universal constant.
\end{prop}
\begin{proof}
For every $\tilde{x}=(x,z)\in \{-1,1\}^{2n}$ we define the $n$-dimensional states:
\begin{align*}
|\varphi_{x}\rangle=\frac{1}{\sqrt{n}} \sum_{i=1}^n x_i |i\rangle  \text{     }\text{    and   }\text{     }|\varphi_z\rangle=\frac{1}{\sqrt{n}} \sum_{j=1}^n z_j |j\rangle.
\end{align*} Then, we consider the operator $\rho_{\tilde{x}}=|\varphi_x\rangle \langle \varphi_z|$ with trace norm equal one, which is, in general, not self-adjoint.

Also, for every $y\in \{-1,1\}^{n^2}$ we consider the  (non self-adjoint) operator $A_y$ in $M_n$, whose matrix in the canonical basis is $(y_{i,j})_{i,j=1}^n=(r_{i,j}(y))_{i,j=1}^n$ and then its normalized version $$\tilde{A}_y=\frac{A_y}{\|A_y\|_{M_n}}.$$

According to Lemma \ref{l:quantumvalue}, up to the fact that neither the $\rho_{\tilde{x}}$'s nor the $\tilde{A}_y$'s are self-adjoint elements,  it suffices to lower bound the quantity
\begin{align*}
\Big|\sum_{\tilde{x}\in \setx,y\in \sety}T_{\tilde{x},y}tr\big(\tilde{A}_y  \rho_{\tilde{x}}\big)\Big|.
\end{align*}

To this end, write
\begin{align*}
\Big|\sum_{\tilde{x}\in \tilde{\setx},y\in \sety}T_{\tilde{x},y}tr\big(\tilde{A}_y  \rho_{\tilde{x}}\big)\Big|=\sum_{(x,z),y} T_{(x,z),y} tr(|\varphi_x\rangle \langle \varphi_z| \tilde{A}_y)=\frac{1}{M}\frac{1}{n}\sum_{(x,z),y} \sum_{i,j} x_i z_j y_{i,j} \frac{ \sum_{i',j'} x_{i'} y_{i',j'} z_{j'}}{\|A_y\|_{M_n}}. 
\end{align*}

Using the orthogonality properties of the Rademacher functions we see that $\sum_{x\in \{-1,1\}^{n}} x_i x_{i'}=\delta_{i,i'}2^n$ and  $\sum_{z\in \{-1, 1\}^{n}} z_j z_{j'}=\delta_{j,j'}2^n$. Since $y_{i,j}^2=r_{i,j}(y)^2=1$ for every $y$ and $i,j$, we deduce that %
 \begin{align*} 
 \Big|\sum_{\tilde{x}\in \tilde{\setx},y\in \sety}T_{\tilde{x},y}tr\big(\tilde{A}_y  \rho_{\tilde{x}}\big)\Big|&
 =\frac{2^{2n}}{Mn}n^2\sum_{y}  \frac{1}{\|A_y\|_{M_n}} \geq C\frac{2^{2n}}{M}n2^{n^2}\frac{1}{\sqrt {n}}=C\frac{2^{2n+n^2}}{M}\sqrt{n}.
 \end{align*}In the last inequality we have used Jensen's inequality to conclude that $$\sum_{y\in \{-1,1\}^{n^2}} \frac{1}{\|A_y\|_{M_n}}= 2^{n^2}\sum_{y\in \{-1,1\}^{n^2}} \frac{1}{2^{n^2}}\frac{1}{\|A_y\|_{M_n}}=2^{n^2}\mathbb E\frac{1}{\|A_y\|_{M_n}}\geq 2^{n^2}\frac{1}{\mathbb E\|A_y\|_{M_n}},$$together with the well known theorem (\cite[Theorem 2]{Latala}) which states that 
 \begin{align}\label{chevet}
 \mathbb E\|A_y\|_{M_n}\leq C' \sqrt{n}
 \end{align} for a certain universal constant $C'$.
 
 Since  by Lemma \ref{control of M}, $M\leq n2^{n^2+2n}$, we conclude that 
 \begin{align}\label{final inequality} 
 \omega^*_{o.w.-\log n}(G)\geq \Big|\sum_{\tilde{x}\in \tilde{\setx},y\in \sety}T_{\tilde{x},y}tr\big(A_y  \rho_{\tilde{x}}\big)\Big|\geq \frac{C}{\sqrt{n}}.
 \end{align}
 
 Finally, by splitting each $\rho_{\tilde{x}}$ and each $\tilde{A}_y$ in their real and imaginary part (which are both self-adjoint) we easily conclude that Eq. (\ref{final inequality}) holds for some families of self-adjoint operators $(\rho_{\tilde{x}})_{\tilde{x}}$ and $(A_y)_{y}$ verifying $\|\rho_{\tilde{x}}\|_{S_1^n}\leq 1$ and $\|A_y\|_{M_n}\leq 1$ for every $\tilde{x}$, $y$ at the price of replacing $C$ by $C/4$.
\end{proof}
\section{Communication values of a XOR game as  tensor norms}\label{S:tensor norms}

In this section we explain how the different values of XOR games that we consider in this work can be understood in terms of tensor norms on Banach spaces. The use of tensor norms has been very useful in the context of Bell inequalities violations (see \cite[Section 2]{PV-Survey} for more information about this) and in the context of communication complexity (see for instance \cite{LMSS}, \cite{LS}). In section \ref{S:atom reduction} we will show how the abstract point of view of tensor norms allows us to apply a method to reduce the number of inputs in nonlocal games developed in \cite{JOP16} to prove Theorem \ref{Thm:MainIV}.

Let us start by introducing some standard notation in Banach space theory. 

Given a normed space $X$, let us denote by $B_X=\{x\in X: \|x\|\leq 1\}$ its unit ball and by $X^*$ its dual space. This space consists of all linear and continuous maps from $X$ to the scalar field $\mathbb K$ ($\mathbb R$ or $\mathbb C$) and has a natural norm defined as $\|x^*\|_{X^*}=\sup_{x\in B_X}|\langle x^*, x\rangle|$.

Given two finite dimensional normed spaces $X$, $Y$, the space of linear maps from $X$ to $Y$, $L(X,Y)$, is naturally identified with the algebraic tensor product $X^*\otimes Y$. The correspondence goes as follows: Given any element  $u=\sum_{i}x_i^*\otimes y_i\in X^*\otimes Y$, we define the linear map $T_u\in L(X,Y)$ by $T_u(x)=\sum_{i}\langle x_i^*,x\rangle y_i$. In addition, if we fix bi-orthogonal basis $(x_i)_i$ of $X$ and $(x_i^*)_i$ of $X^*$, to any linear map $T\in L(X,Y)$ we can associate the tensor $u_T=\sum_i x_i^*\otimes  T(x_i)\in X^*\otimes Y$. Note that $T_{u_T}=T$.

If we want to make the previous identification isometric, we must introduce some norms on the spaces $L(X,Y)$ and $X^*\otimes Y$. The space $L(X,Y)$ is  naturally endowed with the (operator) norm 
\begin{align}\label{norm operator}
\|T\|=\sup_{x\in B_X}\|T(x)\|_{Y}.
\end{align}We will denote by $\mathcal L(X,Y)$ the space $L(X,Y)$  endowed with the previous norm.

The tensor product of two normed spaces can be endowed with different tensor norms. In this paper, we will be  interested in the so called \emph{$\epsilon$-norm} which we define next: 

Given  two finite dimensional normed spaces $X$, $Y$, we will consider the \emph{$\epsilon$-norm} of $u\in X\otimes Y$, defined by
\begin{align}\label{injective norm}
\|u\|_{X\otimes_\epsilon Y}=\sup_{x^*\in B_{X^*},\,y^*\in B_{Y^*}}\big|\langle u, x^*\otimes y^*\rangle\big|,
\end{align}where for a given $u=\sum_i x_i\otimes y_i \in X\otimes Y$, $\langle u, x^*\otimes y^*\rangle=\sum_i\langle x^*, x_i\rangle \langle y^*, y_i\rangle$. We will denote $X\otimes_{\epsilon} Y$ the space $X\otimes Y$ with the previous norm. Considering finite dimensional spaces is not relevant in these definitions, but it is the only case that we will use in this work and it allows us to ignore the completion of the normed spaces.

It is very easy to see from (\ref{norm operator}) and (\ref{injective norm}) that $\|u\|_{X\otimes_\epsilon Y}=\|T_u:X^*\rightarrow Y\|$ so that we have the following isometric identification:
\begin{align}\label{inj-op norm}
\mathcal L(X,Y)=X^*\otimes_{\epsilon} Y.
\end{align}

We will only need two properties about the $\epsilon$-norm, which can be easily proved from its definition. The first one is usually called the metric mapping property and it states (see \cite[pag. 46]{DF} for a proof) that for any spaces $X$, $Y$, $Z$, $W$ and any linear maps $T:X\lra Z$ and $S:Y\lra W$ one has 
\begin{align}\label{MPP}
\|T\otimes S:X\otimes_{\epsilon} Y\lra Z\otimes_{\epsilon} W\|=\|T\|\|S\|.
\end{align}

In Section \ref{S:atom reduction} we will also need the  injectivity of the $\epsilon$-norm (see \cite[pag. 49]{DF} for a proof) : for any subspaces $Z\subset X$ and $W\subset Y$ the $\epsilon$-norm in   $Z\otimes W$ is just the restriction of the $\epsilon$-norm in  $X\otimes Y$. Schematically, 
\begin{align}\label{injectivity}
Z\otimes_{\epsilon} W\subset X\otimes_{\epsilon} Y.
\end{align}

When working with XOR games, we are interested in the space $\ell_1^N:=(\mathbb R^N, \|\cdot \|_1)$, where $\|(x_i)_i \|_1=\sum_{i=1}^N|x_i|$ and its dual space $(\ell_1^N)^*=\ell_\infty^N:=(\mathbb R^N, \|\cdot \|_\infty)$, where $\|(x_i)_i \|_\infty=\sup_{i=1}^N|x_i|$. Note that this is the case $p=1$ in (\ref{Def p-norm}) for a general $N$ and the duality relation (\ref{duality relation}) in this case.

We consider  a XOR game $G$, whose coefficients are $T_{x,y}=\pi(x,y) f(x,y)$ for every $x\in \setx$, $y\in \sety$. We will view this game as an element  $$G=\sum_{x\in \setx, y\in \sety}T_{x,y}e_x\otimes e_y\in \ell_1^{\cardx}\otimes \ell_1^{\cardy},$$
equivalently as an element  $$T_{G}:\ell_\infty^{\cardx}\longrightarrow \ell_1^{\cardy}.$$

It is easy to prove (see \cite[Equations (10) and (18)]{PV-Survey}) that given a XOR game $G$, the value $\omega(G)$ defined in (\ref{classical bias XOR}) can be written as
$$\omega(G)=\|G\|_{\ell_1^{\cardx}\otimes_\epsilon \ell_1^{\cardy}}=\|T_{G}:\ell_\infty^{\cardx}\lra \ell_1^{\cardy}\|,$$where the second equality follows from (\ref{inj-op norm}).

Note that the coefficients $(T_{x,y})_{x,y}$ verify the normalization condition
\begin{align}\label{NSG condition}
\sum_{x\in \setx, y\in \sety}|T_{x,y}|=1.
\end{align}

Conversely, any real matrix $(T_{x,y})_{x,y=1}^{\cardx,\cardy}$ verifying condition (\ref{NSG condition}) corresponds to the coefficients of a  XOR game. Indeed, we just need to consider the game defined by the probability distribution $\pi(x,y)=|T_{x,y}|$ and the function $f(x,y)=\sign(T_{x,y})$ for every  $x\in \setx$, $y\in \sety$. Hence, XOR games can be identified with matrices $(T_{x,y})_{x,y=1}^{\cardx,\cardy}$ verifying condition (\ref{NSG condition}). 

In this work, we are not interested in the classical value of  $G$, but in the values $\omega_{o.w-c}(G)$ and $\omega^*_{o.w-c}(G)$. We show next that these values can be easily expressed as certain norms of naturally associated linear maps. 

For any natural number $d$ we  consider the element $$G \otimes id_{\ell_1^d}:=\sum_{x\in \setx, y\in \sety}\sum_{i=1}^d T_{x,y}(e_x\otimes e_i)\otimes (e_y\otimes e_i)\in \ell_1^{\cardx}(\ell_\infty^d)\otimes \ell_1^{\cardy}(\ell_1^d).$$

Here, for any normed space $X$ and any natural number $N$ we denote by $\ell_1^N(X)$ the space $\{(x_i)_{i=1}^N: x_i\in X, i:1,\cdots, N\}$ endowed with the norm $$\|(x_i)_i \|_{\ell_1^N(X)}=\sum_{i=1}^N\|x_i\|_X.$$It is easy to check that $\ell_1^N(\ell_1^d)=\ell_1^{Nd}$ and standard calculations show  that the dual space of $\ell_1^N(X)$ is the space $\ell_\infty^N(X^*)$, defined as $\{(y_i)_{i=1}^N: y_i\in X^*, i:1,\cdots, N\}$ endowed with the norm $$\|(y_i)_i \|_{\ell_\infty^N(X)}=\sup_{i=1,\cdots ,N}\|y_i\|_{X^*}.$$

The following result is straightforward from Lemma \ref{Lemma: classical} and the definitions above.
\begin{lemma}\label{connection norms I}
Let $G$ be a XOR game with coefficients $(T_{x,y})_{x,y}$. We identify $G$ with a tensor $G=\sum_{x\in \setx, y\in \sety}T_{x,y}e_x\otimes e_y\in \ell_1^{\cardx}\otimes \ell_1^{\cardy}$ and the corresponding operator $T_{G}:\ell_\infty^{\cardx}\lra  \ell_1^{\cardy}$. Then, 
$$\omega_{o.w-\log d}(G)=\|G\otimes id_{\ell_1^d}\|_{\ell_1^{\cardx}(\ell_\infty^d)\pe \ell_1^{\cardy}(\ell_1^d)}=
\|T_{G}\otimes id:\ell_\infty^{\cardx}(\ell_1^{d})\lra \ell_1^{\cardy}(\ell_1^{d})\|.$$
\end{lemma}

In a completely analogous way we can describe the value $\omega_{o.w-c}^*(G)$. Let us denote by $S_\infty^d$ the space of $d\times d$-complex matrices with the operator norm (note that we denoted this space by $M_d$ in Section \ref{S: communication}) and $S_1^d$,  the space of $d\times d$ complex matrices with the trace norm. It is well known that $(S_\infty^d)^*=S_1^d$ in analogy with the duality between $(\ell_\infty^{d})^*$ and $\ell_1^{d}$. Then, we can regard the element $G$ as an element $$G \otimes id_{S_1^d}:=\sum_{x\in \setx, y\in \sety}\sum_{i,j=1}^d T_{x,y}(e_x\otimes e_{i,j} )\otimes (e_y\otimes e_{i,j})\in \ell_1^{\cardx}(S_\infty^d)\otimes \ell_1^{\cardy}(S_1^d),$$where here $e_{i,j}$ denotes the matrix with all entries equal zero up to a one at the entry $(i,j)$.

The following result is straightforward from Lemma \ref{l:quantumvalue} and the definitions above.
\begin{lemma}\label{connection norms II}
Let $G$ be a XOR game with coefficients $(T_{x,y})_{x,y}$. We identify $G$ with a tensor $G=\sum_{x\in \setx, y\in \sety}T_{x,y}e_x\otimes e_y\in \ell_1^{\cardx}\otimes \ell_1^{\cardy}$ and the corresponding operator $T_{G}:\ell_\infty^{\cardx}\lra  \ell_1^{\cardy}$. Then, 
$$\frac{1}{4}\|G\otimes id_{S_1^d}\|_{\ell_1^{\cardx}(S_\infty^d)\pe \ell_1^{\cardy}(S_1^d)}\leq \omega_{o.w-\log d}^*(G)\leq \|G\otimes id_{S_1^d}\|_{\ell_1^{\cardx}(S_\infty^d)\pe \ell_1^{\cardy}(S_1^d)}.$$
\end{lemma}Analogously to Lemma \ref{connection norms I}, we can write the norm $
\|T_{G}\otimes id_{S_1^d}:\ell_\infty^{\cardx}(S_1^{d})\lra \ell_1^{\cardy}(S_1^{d})\|$ instead of $\|G\otimes id_{S_1^d}\|_{\ell_1^{\cardx}(S_\infty^d)\pe \ell_1^{\cardy}(S_1^d)}$ in the statement of the previous lemma.

The reason why we do not have an equality in Lemma \ref{connection norms II} is that in Lemma  \ref{l:quantumvalue}  the optimization is over self-adjoint matrices while the definition of the previous norms consider general complex matrices. However, if in the definition of the previous norms we restrict to the space of self-adjoint $d\times d$-complex matrices with the operator norm $(S_\infty^d)^{s.a}$ and to the space of self-adjoint $d\times d$-complex matrices with the trace norm $(S_1^d)^{s.a}$, then we can conclude that $$\omega_{o.w-\log d}^*(G)= \|G\otimes id_{S_1^d}\|_{\ell_1^{\cardx}((S_\infty^d)^{s.a})\pe \ell_1^{\cardy}((S_1^d)^{s.a})}.$$

The following result is trivial from Lemma \ref{connection norms I} and Lemma \ref{connection norms II}
\begin{theorem}\label{theorem quotient norms}
Let $G$ be a XOR game. Then,
\begin{align}
\frac{\omega^*_{o.w-\log d}(G)}{\omega_{o.w-\log k}(G)}\leq \frac{\|G\otimes id_{S_1^d}\|_{\ell_1^{\cardx}(S_\infty^d)\otimes_{\epsilon} \ell_1^{\cardy}(S_1^d)}}{\|G\otimes id_{\ell_1^k}\|_{\ell_1^{\cardx}(\ell_\infty^k)\otimes_{\epsilon} \ell_1^{\cardy}(\ell_1^k)}}\leq 4\frac{\omega^*_{o.w-\log d}(G)}{\omega_{o.w-\log k}(G)}.
\end{align}
\end{theorem}

It is important to mention that, although we have written the paper removing all Banach space terminology so that no knowledge of it is needed to understand the proofs, both Theorem \ref{Thm:MainI} and Theorem \ref{Thm:MainIII} were thought in terms of tensor norms and factorizations of operators by looking at the relation stated in Theorem \ref{theorem quotient norms}. As we will see next, this point of view is crucial in the proof of Theorem \ref{Thm:MainIV}. 
\section{Reducing the number of inputs of the game $G$}\label{S:atom reduction}

In this section we modify the game $G$ introduced in Section \ref{S:values} to obtain a new game $\tilde{G}$ defined with inputs $x, y\in \{1,\cdots, m\}$, with $m\leq cn^8$ for a given universal constant $c$, and such that $\tilde{G}$  verifies 
\begin{align*}
\frac{\omega^*_{o.w-\log n}(\tilde{G})}{\omega_{o.w-\log n}(\tilde{G})}\geq C\frac{\sqrt{n}}{\log n},
\end{align*}where $C$ is a universal constant.

This will prove Theorem \ref{Thm:MainIV}. In addition, using Theorem \ref{Thm:MainI}, Corollary \ref{main cor} follows straightforward.

According to Theorem \ref{theorem quotient norms}, Theorem \ref{Thm:MainIV} will follow from the following result.
\begin{theorem}\label{Theorem 3 with norms}
There exists a family of XOR games $(\tilde{G}_n)_n$ such that for every $n$, $\tilde{G}_n$ has $m_n \leq c n^8$ inputs per player and it verifies
\begin{align}\label{desired property}
\frac{\|\tilde{G}\otimes id_{S_1^n}\|_{\ell_1^{m}(S_\infty^n)\otimes_{\epsilon} \ell_1^{m}(S_1^n)}}{\|\tilde{G}\otimes id_{\ell_1^n}\|_{\ell_1^{m}(\ell_\infty^n)\otimes_{\epsilon} \ell_1^{m}(\ell_1^n)}}\geq D\frac{\sqrt{n}}{\log n}.
\end{align}

Here, $D$ and $c$ are universal constants.
\end{theorem}

As we did in Section \ref{S:values}, we will not write the dependence of $n$ in the notation.

The key point in this proof is to understand the game $G$ introduced in Section \ref{S:values} as a tensor in $\ell_1^{2^{2n}}\otimes \ell_1^{2^{n^2}}$. Then, we will prove our result as a direct consequence of the recent work \cite{JOP16}, where a method to reduce the number of inputs of certain games was studied.

In fact, according to the identification explained in Section \ref{S:tensor norms} the reader will find easy to check that our game $G$ defined in Section \ref{S:values} corresponds to the tensor $$G=\frac{1}{M}\sum_{\substack{x,z\in \{-1,1\}^n\\y\in \{-1,1\}^{n^2}}}\Big(\sum_{i,j=1}^n x_iz_jy_{i,j}\Big)e_{x,z}\otimes e_y\in \ell_1^{2^{2n}}\otimes \ell_1^{2^{n^2}},$$where $M$ was introduced in (\ref{Def M}).

The element $G$ can be understood in the following way: Let us define the linear maps $j_1:\ell_2^{n^2}\lra \ell_1^{2^{2n}}$ and $j_2:\ell_2^{n^2}\lra \ell_1^{2^{n^2}}$ as
\begin{align}\label{maps j_1-j_2}
j_1(e_{i,j})=\sum_{x,z\in \{-1,1\}^n}x_iz_je_{x,z},\text{     } \text{     } \text{   and   }\text{     }\text{     }j_2(e_{i,j})=\sum_{y\in \{-1,1\}^{n^2}}y_{i,j}e_{y}
\end{align}for every $i,j=1,\cdots, n$. Then, it is easy to check that $$G=(j_1\otimes j_2)(I),$$where $$I=\frac{1}{M}\sum_{i,j=1}^ne_{i,j}\otimes e_{i,j}\in \ell_2^{n^2}\otimes \ell_2^{n^2}.$$

This way of writing the game $G$ shows that, although the space $\ell_1^{2^{2n}}\otimes \ell_1^{2^{n^2}}$ has very large dimension, the relevant space to define our element $G$ is $j_1(\ell_2^{n^2})\otimes j_2(\ell_2^{n^2})\subset \ell_1^{2^{2n}}\otimes \ell_1^{2^{n^2}}$, which has a much lower dimension. This situation allows us to apply the empirical method studied in \cite{JOP16}. More precisely, let us state \cite[Proposition 3.1]{JOP16} in the more explicit form explained in \cite[Remark 3.1]{JOP16}:
\begin{prop}\label{Proposition reduction}
Let $X$ be a Banach space, $N$ be a natural number and $E\subset \ell_1^N(X)$ be a $d$-dimensional subspace of $\ell_1^N(X)$. Then, for any $\epsilon>0$ and taking $m$ to be the integer part of $C(\vr) d^2$, there exists a map $J:\ell_1^N\rightarrow \ell_1^m$ such that $J\otimes id_X$ defines a $(1+\epsilon)$-isomorphism from $E$ to $\ell_1^m(X)$. Here, one can take $C(\vr) = C_0 \vr^{-2} \log(\vr^{-1})$ for a universal constant $C_0$.
\end{prop}

Here, the fact that $J\otimes id_X$ is a $(1+\epsilon)$-isomorphism from $E$ to $\ell_1^m(X)$ means that for every $e\in E$ we have
$$\frac{1}{\sqrt{1+\epsilon}}\|e\|_{ \ell_1^N(X)}\leq \|(J\otimes id_X)(e)\|_{\ell_1^m(X)}\leq \sqrt{1+\epsilon}\|e\|_{ \ell_1^N(X)}.$$

The proof of Proposition \ref{Proposition reduction} is based on the phenomenon of concentration of measure (and the ideas developed in \cite{Sch87}) and does not provide an explicit map $J$. However, it guarantees the existence of such a map and, in addition, that $J$ is defined via a (random) choice of indices $i_1,\cdots, i_m\in \{1,\cdots, N\}$ and some positive numbers $\alpha_1,\cdots, \alpha_m$ such that 
\begin{align}\label{J-map}
J(x_1,\cdots, x_N)= (\alpha_1x_{i_1},\cdots, \alpha_mx_{i_m})\text{      }\text{      }\text{   for every    }\text{      }\text{      }(x_1,\cdots, x_N)\in \ell_1^N.
\end{align}

Proposition \ref{Proposition reduction} encourages us to look at our game $G$ in the following way. Let us denote $X=j_1(\ell_2^{n^2})\subset \ell_1^{2^{2n}}$ and $Y=j_2(\ell_2^{n^2})\subset \ell_1^{2^{n^2}}$, where the maps $j_1$ and $j_2$ were defined  in (\ref{maps j_1-j_2}).
Moreover, we will consider the spaces $\tilde{X}=X\otimes S_\infty^n$ and $\tilde{Y}=Y\otimes S_1^n$ endowed with the norms inherit by the inclusions $\tilde{X}\subset \ell_1^{2^{2n}}(S_\infty^n)$ and $\tilde{Y}\subset \ell_1^{2^{n^2}}(S_1^n)$. Note that $$G\otimes id_{S_1^n}\in \tilde{X}\otimes \tilde{Y}.$$

According to Proposition \ref{Proposition reduction} applied to, say $\epsilon=1/2$, since $\tilde{X}$ (resp. $\tilde{Y}$) is a $n^4$-dimensional subspace of $\ell_1^{2^{2n}}(S_\infty^n)$ (resp. $\ell_1^{2^{n^2}}(S_1^n)$), there exists a map $J_1:\ell_1^{2^{2n}}\lra \ell_1^m$ (resp. $J_2:\ell_1^{2^{n^2}}\lra \ell_1^m$) with $m\leq cn^8$ for a certain universal constant $c$ , such that $J_1\otimes id_{S_\infty^n}$ (resp. $J_2\otimes id_{S_1^n}$) is a $1/2$-embedding on $\tilde{X}$ (resp. $\tilde{Y}$). 

Note that, without loss of generality, we can assume that $\|J_1\otimes id_{S_\infty^n}\|\leq 3/2$ and $\|J_1^{-1}\otimes id_{S_\infty^n}\|\leq 1$.  Analogously, we can assume that $\|J_2\otimes id_{S_1^n}\|\leq 3/2$ and $\|J_2^{-1}\otimes id_{S_1^n}\|\leq 1$. In addition, since $\ell_\infty^n$ can be identified with the space of diagonal matrices in $S_\infty^n$ and $\ell_1^n$ can be identified with the space of diagonal matrices in $S_1^n$, the previous estimates imply that  $\|J_1\otimes id_{\ell_\infty^n}\|\leq 3/2$ and $\|J_2\otimes id_{\ell_1^n}\|\leq 3/2$.

In order to define our new game $\tilde{G}$, we will consider the tensor $$(J_1\otimes J_2)(G)\in \ell_1^m\otimes \ell_1^m.$$In fact, we must consider the normalization (\ref{NSG condition}) of this element. That is, if we write $(J_1\otimes J_2)(G)=\sum_{x,y=1}^mH_{x,y}e_x\otimes e_y$, and denote
\begin{align}\label{normalization new game}
N=\sum_{x,y=1}^m|H_{x,y}|,
\end{align}our element will be defined as $$\tilde{G}=\frac{1}{N}\sum_{x,y=1}^mH_{x,y}e_x\otimes e_y.$$

Note that although the maps $J_1$ and $J_2$ are not explicit, their simple form (\ref{J-map}) allows us to have a clear idea about the form of the new game $\tilde{G}$. Indeed, according to (\ref{J-map}) the new game $\tilde{G}$ is defined by removing from $G$ all but $m$ inputs for Alice $x_1,\cdots , x_m$ and all but $m$ inputs for Bob $y_1,\cdots , y_m$. Then, the value of the function $f(x_i,y_j)$ for these inputs will be exactly the same as for the game $G$, but the new probability distribution is of the form $\tilde{\pi}(x_i,y_j)=\pi(x_i,y_j)\alpha_i\beta_j(1/N)$ for some (non-explicit) positive numbers $\alpha_i$ and $\beta_j$ and the number $N$ is given by (\ref{normalization new game}).

We are now ready to prove the main result of this section.
\begin{proof}[Proof of Theorem \ref{Theorem 3 with norms}]
We will see that the game $\tilde{G}$ defined above verifies what we want. First of all, we have already explained that it has $m\leq cn^8$ inputs per player. Hence, it suffices to show the following estimates:
\begin{align}\label{estimate I}
&\|\tilde{G}\otimes id_{\ell_1^{n}}\|_{\ell_1^{m}(\ell_\infty^{n})\otimes_{\epsilon} \ell_1^{m}(\ell_1^n)}\leq \frac{9}{4}\|G\otimes id_{\ell_1^{n}}\|_{\ell_1^{2^{2n}}(\ell_\infty^{n})\otimes_{\epsilon} \ell_1^{2^{n^2}}(\ell_1^{n})}, \text{   }\text{  and }\\\label{estimate II}&\|\tilde{G}\otimes id_{S_1^{n}}\|_{\ell_1^{m}(S_\infty^{n})\otimes_{\epsilon} \ell_1^{m}(S_1^{n})}\geq \|G\otimes id_{S_1^{n}}\|_{\ell_1^{2^{2n}}(S_\infty^{n})\otimes_{\epsilon} \ell_1^{2^{n^2}}(S_1^{n})}. 
\end{align}

With these estimates at hand, the result follows straightforward since we can conclude
\begin{align*}
\frac{\|\tilde{G}\otimes id_{S_1^{n}}\|_{\ell_1^{m}(S_\infty^{n})\otimes_{\epsilon} \ell_1^{m}(S_1^{n})}}{\|\tilde{G}\otimes id_{\ell_1^{n}}\|_{\ell_1^{m}(\ell_\infty^{n})\otimes_{\epsilon} \ell_1^{m}(\ell_1^n)}}\geq \frac{4}{9}\frac{\|G\otimes id_{S_1^{n}}\|_{\ell_1^{2^{2n}}(S_\infty^{n})\otimes_{\epsilon} \ell_1^{2^{n^2}}(S_1^{n})}}{\|G\otimes id_{\ell_1^{n}}\|_{\ell_1^{2^{2n}}(\ell_\infty^{n})\otimes_{\epsilon} \ell_1^{2^{n^2}}(\ell_1^{n})}}\geq D\frac{\sqrt{n}}{\log n},
\end{align*}for a certain universal constant $D$. Indeed, the first inequality is a consequence of (\ref{estimate I}) and (\ref{estimate II}), while the second inequality follows by putting together Theorem \ref{Thm:MainIII} and Theorem \ref{theorem quotient norms}.

 In order to see (\ref{estimate I}), write 
\begin{align*}
\|\tilde{G}\otimes id_{\ell_1^{n}}\|_{\ell_1^{m}(\ell_\infty^{n})\otimes_{\epsilon} \ell_1^{m}(\ell_1^{n})}&=\|\big((J_1\otimes d_{\ell_1^{n}})\otimes (J_2\otimes d_{\ell_\infty^{n}})\big)(G\otimes id_{\ell_1^{n}})\|_{\ell_1^{m}(\ell_\infty^{n})\otimes_{\epsilon} \ell_1^{m}(\ell_\infty^{n})}\\&\leq \|J_1\otimes id_{\ell_1^{n}}\| \|J_2\otimes id_{\ell_\infty^{n}}\|\|G\otimes id_{\ell_1^{n}}\|_{\tilde{X}\otimes_{\epsilon} \tilde{Y}}\\&\leq 9/4\|\|G\otimes id_{\ell_1^{n}}\|_{\ell_1^{2^{2n}}(\ell_\infty^{n})\otimes_{\epsilon} \ell_1^{2^{n^2}}(\ell_1^{n})}.
\end{align*}
Here, in the first inequality we used property (\ref{MPP}) and in the last one we used the known upper bounds for $\|J_1\otimes id_{\ell_1^{n}}\|$ and  $\|J_2\otimes id_{\ell_\infty^{n}}\|$ and property (\ref{injectivity}).

For the estimate (\ref{estimate II}), write
\begin{align*}
\|G\otimes id_{S_1^{n}}\|_{\ell_1^{2^{2n}}(S_\infty^{n})\otimes_{\epsilon} \ell_1^{2^{n^2}}(S_1^{n})}&=\|G\otimes id_{S_1^{n}}\|_{\tilde{X}\otimes_{\epsilon} \tilde{Y}}\\&=\|\big((J_1^{-1}\otimes id_{S_1^{n}})\otimes (J_2^{-1}\otimes id_{S_\infty^{n}})\big)(\tilde{G}\otimes id_{S_1^{n}})\|_{\tilde{X}\otimes_{\epsilon} \tilde{Y}}\\&\leq \|(J_1^{-1}\otimes id_{S_1^{n}})\|\|(J_2^{-1}\otimes id_{S_\infty^{n}})\| \| \tilde{G}\otimes id_{S_1^{n}}\|_{\ell_1^{m}(S_\infty^{n})\otimes_{\epsilon} \ell_1^{m}(S_1^{n})}\\&\leq \| \tilde{G}\otimes id_{S_1^{n}}\|_{\ell_1^{m}(S_\infty^{n})\otimes_{\epsilon} \ell_1^{m}(S_1^{n})}.
\end{align*}
Here, in the first equality we used property (\ref{injectivity}) and in the second inequality we used property (\ref{MPP}).
\end{proof}

Note that in order to apply Theorem \ref{Thm:MainI} to obtain Bell inequality violations from Theorem \ref{Thm:MainIV} we would need to have a XOR game for which the quantity $\omega^*_{o.w-\log n}(\tilde{G})$ is larger than $\omega_{o.w-2\log n}(\tilde{G})$ rather than $\omega_{o.w-\log n}(\tilde{G})$ as we have just proved in the previous result. The previous proof can be easily modified to obtain the new estimate. Indeed, we could follow the previous proof step by step by just replacing the spaces $S_1^n$ and $S_\infty^n$ by $S_1^{n^2}$ and $S_\infty^{n^2}$ respectively. The prize to pay for this is that the new spaces $\bar{X}=X\otimes S_\infty^{n^2}$ and $\bar{Y}=Y\otimes S_1^{n^2}$ have dimension $n^6$,  and after applying Proposition \ref{Proposition reduction} we can only assure that our new game $\bar{G}$ will have $m\leq c n^{12}$ inputs per player. This modification allows us to prove the analogous estimates to (\ref{estimate I}) and (\ref{estimate II}):
\begin{align*}
&\|\bar{G}\otimes id_{\ell_1^{n^2}}\|_{\ell_1^{m}(\ell_\infty^{n^2})\otimes_{\epsilon} \ell_1^{m}(\ell_1^{n^2})}\leq \frac{9}{4}\|G\otimes id_{\ell_1^{n^2}}\|_{\ell_1^{2^{2n}}(\ell_\infty^{n^2})\otimes_{\epsilon} \ell_1^{2^{n^2}}(\ell_1^{n^2})}, \text{   }\text{  and }\\&\|\bar{G}\otimes id_{S_1^{n}}\|_{\ell_1^{m}(S_\infty^{n})\otimes_{\epsilon} \ell_1^{m}(S_1^{n})}\geq \|G\otimes id_{S_1^{n}}\|_{\ell_1^{2^{2n}}(S_\infty^{n})\otimes_{\epsilon} \ell_1^{2^{n^2}}(S_1^{n})}.
\end{align*}
Hence, we conclude from Theorem \ref{theorem quotient norms} that our new game verifes
\begin{align*}
\frac{\omega^*_{o.w-\log n}(\bar{G})}{\omega_{o.w-2\log n}(\bar{G})}\geq D'\frac{\sqrt{n}}{\log n}
\end{align*}for a certain  universal constant $D'$.
One immediately obtains Corollary \ref{Corollary last} from this last estimate and Theorem \ref{Thm:MainI}.

\end{document}